\documentclass[]{article}

\title{Technical Report: \\ Globally reasoning about localised security policies in distributed systems}\author{Alejandro Mario Hernandez\\ Department of Informatics, Technical University of Denmark \\ \texttt{aher@imm.dtu.dk}}
\date{}

\usepackage{amsmath,amssymb,enumerate,color,graphicx,subfig,breakurl,xspace}

\newtheorem{theorem}{Theorem}[section]
\newtheorem{lemma}[theorem]{Lemma}

\newcommand{\TrueB}{\textbf{t\!t}\xspace}
\newcommand{\FalseB}{\textbf{f\!f}\xspace}
\newcommand{\TrueS}{{\bf true}}
\newcommand{\FalseS}{{\bf false}}
\newcommand{\textif}{\ \IF\ }

\newcommand{\Four}{\textbf{Four}}
\newcommand{\Aspectbegin}{$$ \left[ \begin{array}{c}}
\newcommand{\Aspectend[1]}{\end{array} \right]#1$$}
\newcommand{\AspectEqbegin}{\left[ \begin{array}{c}}
\newcommand{\AspectEqend[1]}{\end{array} \right]#1}

\newcommand{\Localised}[3]{#1::^{#2}#3}
\newcommand{\LocalisedSubject}[2]{\Localised{\YZl_s}{#1}{#2}}
\newcommand{\LocalisedTarget}[2]{\Localised{\YZl_t}{#1}{#2}}
\newcommand{\w}{w}

\newcommand{\netpar}{\mid\mid}
\newcommand{\val}[1]{\langle\,#1\,\rangle}

\newcommand{\ppar}{\mid}

\newcommand{\OutM}[1]{{\bf out}({#1})}
\newcommand{\InM}[1]{{\bf in}({#1})}
\newcommand{\ReadM}[1]{{\bf read}({#1})}
\newcommand{\locate}[2]{{#1}\,@\,{\it #2}}

\newcommand{\nil}{{\bf 0}}

\newcommand{\YZl}{l}
\newcommand{\Fail}{{\sf fail}}

\newcommand{\Lnt}{\ell^\lambda}
\newcommand{\Ln}{\ell}

\newcommand{\Lc}{\ell}

\newcommand{\Lat}{\ell^\lambda}

\newcommand{\veck}[1]{\overrightarrow{#1}}

\newcommand{\doublebracketleft} {[\![}
\newcommand{\doublebracketright}{]\!]}
\newcommand{\db}[1]{\doublebracketleft #1 \doublebracketright}
\newcommand{\dt}[1]{[\!( #1 )\!]}

\newcommand{\grant}[1]{{\sf grant}({#1})}

\newcommand{\Inference}[2]{\begin{array}{@{}c@{}}#1\\[0em]\hline\\[-0.9em]#2\\
\end{array}}

\newcommand{\IF}{\underline{\sf if}\ }

\begin{document}

\maketitle

\begin{abstract}
In this report, we aim at establishing proper ways for model checking the global security of distributed systems, which are designed consisting of set of localised security policies that enforce specific issues about the security expected.

The systems are formally specified following a syntax, defined in detail in this report, and their behaviour is clearly established by the Semantics, also defined in detail in this report. The systems include the formal attachment of security policies into their locations, whose intended interactions are trapped by the policies, aiming at taking access control decisions of the system, and the Semantics also takes care of this.

Using the Semantics, a Labelled Transition System (LTS) can be induced for every particular system, and over this LTS some model checking tasks could be done. We identify how this LTS is indeed obtained, and propose an alternative way of model checking the not-yet-induced LTS, by using the system design directly. This may lead to over-approximation thereby producing imprecise, though safe, results. We restrict ourselves to finite systems, in the sake of being certain about the decidability of the proposed method.

To illustrate the usefulness and validity of our proposal, we present 2 small case-study-like examples, where we show how the system can be specified, which policies could be added to it, and how to decide if the desired global security property is met.

Finally, an Appendix is given for digging deeply into how a tool for automatically performing this task is being built, including some implementation issues. The tool takes advantage of the proposed method, and given some system and some desired global security property, it safely (i.e. without false positives) ensures satisfaction of it.
\end{abstract}

\begin{section}{Introduction}
When developing distributed systems, security of the information travelling throughout them plays an important role. One must take care of how the information can be accessed in order to meet certain confidentiality requirements. If one can design such systems and formally prove that certain desired security requirements are met by the system, then one has a point even before starting with the implementation of the system itself.

The implementation of such distributed systems is then an issue, as it is supposed to meet the requirements achieved during the specification phase. However, if one can detect possible flaws as early in the development process as possible, those flaws could be more easily overcome. Therefore, a proper specification framework for dealing with security issues in distributed systems is necessary, and as closest as possible the framework can be to a realistic implementation way, helps in keeping things simple while translating from the verified specification to the actual implementation.

In this report, we present a framework with which it is possible to specify distributed systems in a certain level of abstraction. The interactions among the various locations of the system can be predicted and analysed. This allows us to intercept them and perform access control decisions when they are about to happen, in order to provide certain security features. For doing this, the framework provides the possibility of specifying access control security policies scattered among the various locations of the system, and then we aim at proving certain global security properties that might be desirable, which are indeed enforced thanks to the partial contribution done by the localised access control policies.

For achieving this, the combination of the several involved security policies in each possible interaction is needed, and there is a specific Logics that allows us to perform this combination in a consistent way. This is presented in the remainder of this Section.

Later, in Section \ref{sec:AspectKBL} we present the framework for modelling the distributed systems and their localised access control security policies. In Section \ref{sec:ACTLv} we present the Logics for analysing global properties over the framework. In Section \ref{sec:SmartModelChecking} we introduce an alternative approach for performing the model checking, making it faster than in the standard way of inducing the entire state space of the system. We describe some case-study-like examples in Section \ref{sec:examples}, to provide better insights on what the framework can be used for, and finally we conclude in Section \ref{sec:conclusion}. Moreover, there is an extra in Appendix \ref{sec:appendix}, which presents the implementation of a tool for performing the model checking.

\begin{subsection}{A review of Belnap Logics}
\label{subsec:BelnapSynopsis}
For granting access according to some security policy, the traditional boolean values (\TrueB\ and \FalseB) are enough: \TrueB\ grants while \FalseB\ denies access. However, for a distributed setting, where policies might be contradictory (or not sufficiently informative), those two values might not be enough. We shall consider an extension to the Boolean Logic proposed by Belnap~\cite{BelnapOrig}, which has been used for combining security policies~\cite{BrunsHuth08}.

In this extension to the boolean logic, two more values are considered: $\bot$ and $\top$ (read ``bottom'' and ``top''). The traditional \TrueB\ would mean ``the policy accepts the interaction'' whereas the traditional \FalseB\  would mean ``the policy does not accept the interaction''. Since different locations might aim to different security properties, their policies could be contradictory or they may lack information about some particular interaction. These situations can be represented by the two extra values that we have: $\bot$ meaning ``no decision'' and $\top$ meaning ``contradiction''.

With this set of values, which we will call here \textbf{Four} (i.e.~\textbf{Four} = \{$\bot, \TrueB, \FalseB, \top$\}), it is possible to extend the usual boolean operations ($\land$ and $\lor$) and to define new ones ($\otimes$ and $\oplus$). For obtaining that, the set \textbf{Four} is equipped with two partial orderings, say $\leq_k$ and $\leq_t$, as shown in Figure \ref{fig:BelnapFig}.

The usual boolean $\land$ is extended as computing the greatest lower bound in the $\leq_t$ lattice, and the usual $\lor$ as computing the least upper bound (thereby obtaining the same results as in boolean logic in case the operands belong to \{\TrueB, \FalseB\}). Analogously, the new operators over \textbf{Four} can be defined as computing the greatest lower bound (the $\otimes$ operator) and the least upper bound (the $\oplus$ operator), both in the $\leq_k$ lattice.\footnote{Notice that this could also be done by just extending the ``truth tables'' of the usual boolean operators and defining new ones for the new operators That would mean, however, having not just 2 truth tables with 4 cells each but 4 truth tables with 16 cells each, making difficult to remember for a human what each operator produces. Furthermore, it would also make it difficult to \emph{assess} the usefulness of new operators that might be defined.}

\begin{figure}[t]
\begin{center}
\begin{picture}(230,110)
\linethickness{1pt}
\put(40,90){\makebox(20,15){$\top$}}
\put(0,50){\makebox(20,15){\TrueB}}
\put(80,50){\makebox(20,15){\FalseB}}
\put(40,10){\makebox(20,15){$\bot$}}
\put(40,50){\makebox(20,15){$\leq_k$}}
\put(50,22){\line(1,1){35}}
\put(50,22){\line(-1,1){35}}
\put(50,92){\line(1,-1){35}}
\put(50,92){\line(-1,-1){35}}
\put(160,90){\makebox(20,15){\TrueB}}
\put(120,50){\makebox(20,15){$\bot$}}
\put(200,50){\makebox(20,15){$\top$}}
\put(160,10){\makebox(20,15){\FalseB}}
\put(160,50){\makebox(20,15){$\leq_t$}}
\put(170,22){\line(1,1){35}}
\put(170,22){\line(-1,1){35}}
\put(170,92){\line(1,-1){35}}
\put(170,92){\line(-1,-1){35}}
\end{picture}
\end{center}
\caption{The Belnap bilattice {\bf Four}: $\leq_k$ and $\leq_t$.} \label{fig:BelnapFig}
\end{figure}
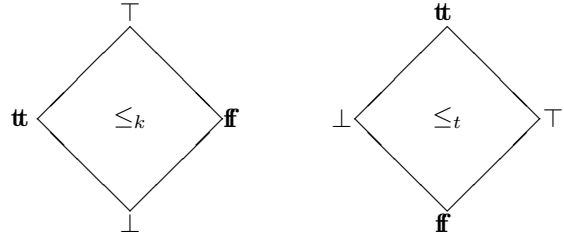

The negation operator $\lnot$ is extended by leaving the two new values unchanged (i.e.~$\lnot \bot = \bot$ and $\lnot \top = \top$), and the implication $\Rightarrow$ is extended as follows:
$$
\begin{array}{lr}
p_1 \Rightarrow p_2 =
\left\{
\begin{array}{ll}
p_2 & \text{if } p_1 \leq_k \TrueB \\
\TrueB & \text{otherwise}
\end{array}
\right.
&
\forall p_1, p_2 \in \textbf{Four}
\end{array}
$$
Another useful operator is the priority $>$, which returns the first operand unless it is $\bot$, in which case it returns the second operand. This would always consider what the first operand suggests unless it has no decision, in which case the second operand is considered.

\begin{subsubsection}{Power of Belnap Logic}
\label{subsubsec:PowerBelnap}
As discussed in several places\cite{BrunsHuth08}\cite{HankinNielsonNielson09}, Belnap Logic is enough for dealing with conflicting policies, thereby is proper for our intended use of combining several policies occurring in different locations of a distributed system.

The decision we shall take is that in order to allow an interaction, no policy should recommend that the interaction should be denied, otherwise we just consider the \emph{negative} policy and deny the interaction. This follows a conservative approach, commonly named as \emph{liberal}, since an interaction is denied as long as there is one policy that suggests so. The proper Belnap operator for combination of policies which could indeed be used for achieving this aim is $\oplus$. Below this shall be clearer.

Another issue that is worth mentioning about Belnap Logic is that it is not capable to express \emph{voting}. Indeed, if one aims at some combination of policies where no conservative approach (or some other approaches with which Belnap Logic could cope) is taken, but some voting among the involved policies is intended, aiming at taking the decision most of the policies agree on, then Belnap Logics is not the right choice.
\end{subsubsection}

\end{subsection}

\end{section}

\begin{section}{A language for distributed security policies}
\label{sec:AspectKBL}
In this Section, a formal language for specifying distributed system is presented. It is based on a subset of the Coordination \cite{coordination} language KLAIM \cite{Klaim}, and it follows a Process Algebraic
approach. Furthermore, the language, named \textbf{AspectKBL}, has primitives for adding localised security policies \cite{GoguenMeseguer82} into the various locations of the distributed system, aiming at providing access control \cite{Sandhu93} decisions based on them.

We will present the \textbf{AspectKBL} language by giving its Syntax and Semantics, and intuitively explaining what can be obtained by using them. The security policies that can be attached to the locations follow also a formal Syntax, which will of course be given as well, and moreover they are then considered by the Semantics, as they are the ones that would either allow the system to progress or not.

In the next Subsections, the Syntax and Semantics follow, and also a small example of a distributed system.

\begin{subsection}{Syntax}
\label{subsec:AspectKBL_Syntax}
The syntax of \textbf{AspectKBL} is given in Tables \ref{tab:syntax_AspKB_NetProcess} and \ref{tab:syntax_AspKB_Policies}, the former being the way of describing basic system locations and the latter the way of describing the attached security policies.

\begin{table}[t]
$$
\begin{array}{l@{\quad}r@{\quad::=\quad}l}

N  \in  {\bf Net}		&	N
	&	N_1\netpar N_2\ \mid\ \Localised{l}{pol}{P}
			\mid\ \Localised{l}{pol}{\val{\veck{\YZl}}}
\\
P  \in  {\bf Proc}		&	P
	&	P_1\ppar P_2\ \mid\ \sum_i a_i.P_i\ \mid\ *P
\\[.5ex]
a  \in  {\bf Act}		&	a
	&	\locate{\OutM{\veck{\ell}}}{\ell}\
			\mid\ \locate{\InM{\veck{\ell^{\lambda}}}}{\ell}
			\mid\ \locate{\ReadM{\veck{\ell^\lambda}}}{\ell}
\\
  {
  c  \in  {\bf Cap} } & {
  c } & {
  \mathbf{out} \mid \mathbf{in} \mid \mathbf{read} }
  \\

\Ln,\Lnt \ \in \ {\bf Loc}	&	\Ln
	&	u \  \mid\ \YZl	\hfill		\Lnt\  ::= \  \Ln \mid\ !u

\end{array}
$$
\caption{{\bf AspectKBL} Syntax -- Nets, Processes, Actions and States.}
\label{tab:syntax_AspKB_NetProcess}
\end{table}
\begin{table*}[t]
$$
\begin{array}{lr@{\quad::=\quad}l}

pol \in {\bf Pol}		&	pol
	&	asp \mid \neg pol \mid pol \oplus pol \mid pol \otimes pol
			\mid pol \Rightarrow pol \mid
\\
				&	\multicolumn{1}{l}{} 
	&	pol > pol \mid pol \wedge pol \mid pol \vee pol \mid \TrueS \mid \FalseS
\\
asp \in {\bf Asp}		&	asp
	&	[rec\ \IF\ cut : cond ]
\\
cut \in {\bf Cut}		&	cut
	&	\Lc :: a^t\,.\, X
\\
a^t \in {\bf Act}^t	&	a^t
	&	\locate{\OutM{\veck{\ell^t}}}{\ell}\
			\mid\ \locate{\InM{\veck{\ell^{t\lambda}}}}{\ell}\
			\mid\ \locate{\ReadM{\veck{\ell^{t\lambda}}}}{\ell}
\\
rec \in {\bf Rec}		&	rec
	&	\ell_1 = \ell_2 \mid \neg rec \mid rec  \oplus  rec  \mid rec  \otimes  rec
			\mid rec \wedge rec \mid
\\
				&	\multicolumn{1}{l}{} 
	&	rec \vee rec \mid rec \Rightarrow rec \mid \TrueS
			\mid \FalseS \mid a\ \hbox{\bf occurs-in}\ X
\\
cond \in {\bf Cond}	&	cond
	&	\ell_1 = \ell_2 \mid \neg cond \mid cond_1  \land  cond_2
			\mid cond_1  \lor  cond_2 \mid
\\
				&	\multicolumn{1}{l}{} 
	&	\TrueS \mid \FalseS \mid\ a\ \hbox{\bf occurs-in}\ X
\\
				&	\ell ^t
	& \ell \mid \_		\hfill		\ell^{t\lambda}\  ::= \  \Ln^{\lambda} \mid\ \_	\hfill

\end{array}
$$
\caption{{\bf AspectKBL} Syntax - Aspects for Security Policies.}
\label{tab:syntax_AspKB_Policies}
\end{table*}

A system consists of a network which consists of a set of parallel locations, each of which has either a process running or some data on it.
A process is a parallel composition of processes, a choice among a set of processes that follow some action or a replicated process, meaning an arbitrary number of the same.
An action can be an \textbf{out}, which writes some data to the \emph{target location} (the one after the `@' symbol), or it can be an \textbf{in} or a \textbf{read}, both of which gather data from the target location, differing in that the former deletes the data from the source whereas the latter does not.
For notation purposes, we use the curly $\ell$ for unknown locations (although it will never be used for expressing actual networks) and the italic $l$ for fixed ones. $u$ refers to variables that are bound after reading from a target, and then if they are used for expressing actual networks they will be bound at runtime. Non-determinism might be introduced both by writing a choice among a set of processes and also by reading from a target that contains several pieces of data (and binding a variable while doing so).

Attached to each location of the system there is some security policy, which can be either a four-valued Belnap Logic combination of policies or a single aspect. This latter consists of some pointcut $cut$, which may at runtime trap some of the interactions the location may be involved into. Furthermore, there is a boolean applicability condition $cond$ that determines for a given trapped interaction whether the aspect should actually be applied. Finally, there is a four-valued Belnap Logic advice $rec$ (for \emph{recommendation}) that tells if the trapped interaction under the given condition should be actually allowed. The operator \textbf{occurs-in}, used for analysing the future behaviour of the trapped process, is defined in Table \ref{tab:semantics_occurs-in}.

\begin{table}[t]
$$
\begin{array}{r@{\ = \ }l}

a\ \textbf{occurs-in}\ (P_1 \ppar P_2)
	&	(a\ \textbf{occurs-in}\ P_1) \vee (a\ \textbf{occurs-in}\ P_2)
\\
a\ \textbf{occurs-in}\ (\sum_i a_i.P_i)
	&	\bigvee_i(a\ {matches}\ a_i \vee a\ \textbf{occurs-in}\ P_i)
\\
a\ \textbf{occurs-in}\ (*P)
	&	a\ \textbf{occurs-in}\ P

\end{array}
$$
\caption{Continuation analysis operator \textbf{occurs-in}.}
\label{tab:semantics_occurs-in}
\end{table}

\end{subsection}

\begin{subsection}{Semantics}
\label{subsec:AspectKBL_Semantics}
The semantics is given by a reduction relation on nets whose rules are given in Table \ref{tab:semantics_AspKBL_Reaction_Rules}. It makes use of a structural congruence relation on nets, consisting of the usual congruence rules besides those given in Table \ref{tab:semantics_AspKB_Structural_Congruence}. It also makes use of a $match$ operator, defined in Table \ref{tab:semantics_pattern_matching_match_vectors}, for matching input patterns to actual data.

\begin{table*}[t]
$$
\begin{array}{l}

\multicolumn{1}{c}{
	\Inference
		{N_1 \rightarrow^{lab} N_1'}
		{N_1 \netpar N_2 \rightarrow^{lab}  N_1'\netpar N_2}
	\qquad\qquad\qquad
	\Inference
		{N \equiv M \quad M \rightarrow^{lab} M' \quad M' \equiv N'}
		{N \rightarrow^{lab} N'}
	}

\\[3.5ex]

(\LocalisedSubject{pol_s}{\mathbf{read}(\veck\Lat)@l_t.P+ \cdots}) \netpar
	(\LocalisedTarget{pol_t}{\langle \veck\YZl \rangle})
\\	\qquad
	\begin{array}{cl}
	\rightarrow^{l_s:\textbf{r}(\textcolor{black}{\veck{l}})@l_t}
		&
		\LocalisedSubject{pol_s}{P\theta} \netpar
			\LocalisedTarget{pol_t}{\langle \veck\YZl \rangle}
	\\[1ex]
	\multicolumn{2}{l}{\textrm{if } \grant{\db{pol_s \oplus pol_t}
		(l_s::\mathbf{read}(\veck\Lat)@l_t.P)}}
	\\
	\multicolumn{2}{l}{\textrm{and } match(\veck\Lat;\veck\YZl)= \theta;}
		
	\end{array}

\\[7.5ex]

(\LocalisedSubject{pol_s}{\mathbf{in}(\veck\Lat)@l_t.P+ \cdots}) \netpar
	(\LocalisedTarget{pol_t}{\langle \veck\YZl \rangle})
\\	\qquad
	\begin{array}{cl}
	\rightarrow^{l_s:\textbf{i}(\textcolor{black}{\veck{l}})@l_t}
		&
		\LocalisedSubject{pol_s}{P\theta}
	\\[1ex]
	\multicolumn{2}{l}{\textrm{if } \grant{\db{pol_s \oplus pol_t}
		(l_s::\mathbf{in}(\veck\Lat)@l_t.P)}}
	\\
	\multicolumn{2}{l}{\textrm{and } match(\veck\Lat;\veck\YZl)= \theta;}
	\end{array}

\\[7.5ex]

(\LocalisedSubject{pol_s}{\mathbf{out}(\veck\YZl)@l_t.P+\cdots}) \netpar
	(\LocalisedTarget{pol_t}{Q})
\\	\qquad
	\begin{array}{cl}
	\rightarrow^{l_s:\textbf{o}(\veck{l})@l_t}
		&
		\LocalisedSubject{pol_s}{P} \netpar \LocalisedTarget{pol_t}{\langle \veck\YZl \rangle}
			\netpar \LocalisedTarget{pol_t}{Q}
	\\[1ex]
	\multicolumn{2}{l}{\textrm{if } \grant{\db{pol_s \oplus pol_t}(l_s::\mathbf{out}
		(\veck\YZl)@l_t.P)};}
	\end{array}

\end{array}
$$
\caption{Reaction Semantics of {\bf AspectKBL} .}
\label{tab:semantics_AspKBL_Reaction_Rules}
\end{table*}
\begin{table}[t]
$$
\begin{array}{lr}

\begin{array}{r@{\quad\equiv\quad}l}
l::^{pol} P_1 \ppar P_2	&	l::^{pol}P_1 \netpar l::^{pol}P_2
\\[1ex]
l ::^{pol}\ * P			&	l ::^{pol}\ P \ppar\ * P
\\[1ex]
l ::^{pol}\ P			&	l ::^{pol} P \netpar l ::^{pol} {\bf 0}
\end{array}

&	\hspace{2cm}

\begin{array}{rcl}
\multicolumn{3}{c}
	{
	\Inference
		{N_1 \equiv N_2}
		{N \netpar N_1 \equiv N \netpar N_2}
	}
\end{array}

\end{array}
$$
\caption{Structural Congruence.}
\label{tab:semantics_AspKB_Structural_Congruence}
\end{table}
\begin{table}[t]
$$
\begin{array}{r@{\quad =\quad }l}
match(\,!u,\veck{\ell^\lambda}\,;\,l,\veck{l}\,)
	&	[l/u] \circ match(\,\veck{\ell^\lambda}\,;\,\veck{l}\,)
\\
match(\,l,\veck{\ell^\lambda}\,;\,l,\veck{l}\,)
	&	match(\,\veck{\ell^\lambda}\,;\,\veck{l}\,)
\\
match(\,\epsilon\,;\,\epsilon\,)
	&	id 
\\
match(\,\cdot\,;\,\cdot\,)
	&	\Fail \qquad\textrm{ otherwise}
\end{array}
$$
\caption{Matching Input Patterns to Data.}
\label{tab:semantics_pattern_matching_match_vectors}
\end{table}

The rules should be straightforwardly understood since they follow the traditional pattern for process calculi, in this case also considering the security policies attached to the locations involved in the interaction. For this purpose, each rule only defines a transition if the policies agree on allowing the interaction to take place. That is the purpose of calling the auxiliary function \grant{}, with the four-valued Belnap Logic combination of the involved policies and the intended action, for turning the four-valued policies' recommendations into an actual boolean decision. This is to be further explained in the Subsection \ref{subsec:AspectKBL_MeaningPolicies}.

When the interaction is actually allowed by the policies, the data is either written to the target location (if the action is an \textbf{out}) or gathered from there (if the action is an \textbf{in} or \textbf{read}, deleting the data in the former case). In the case of reading some data, the continuation process is substituted with it, according to the result of the function $match$ (if there is such result, otherwise the interaction cannot take place). In the case of writing some data, the target location is annotated with the security policy previously existent on it, that is why for the interaction to take place some process running there must be pattern matched.

The syntax of the labels is given by $l:c(\veck{l})@l$, where $c \in$ \{ \textbf{r}, \textbf{i}, \textbf{o} \} (the set of capabilities). This means that each transition will be labelled by the actual parameter bound.

Finally, we may notice that the Semantics of Table \ref{tab:semantics_AspKBL_Reaction_Rules} induces a Labelled Transition System (LTS). In such LTS, a denied operation (by the function \grant) does not occur at all. Indeed, the Semantics does not define any transition in the LTS for a denied operation as there is no rule that might say so.
\end{subsection}

\begin{subsection}{Granting access according to policies}
\label{subsec:AspectKBL_MeaningPolicies}

In the reaction rules of Table \ref{tab:semantics_AspKBL_Reaction_Rules}, the process is either terminated (being replaced by a \nil) or continued (writing a data to some location or being substituted with some read data), and this is done according to the result of the function \grant{}.

The function is just a way to map a four-valued Belnap Logic result into a boolean decision, to actually allow the interaction or not. The approach taken is to map to $\TrueB$ in the cases where the four-valued parameter is $\TrueB \in \Four$ or $\bot \in \Four$. The former is obvious, while the latter aims to allow an interaction that is not explicitly forbidden by any policy. If the four-valued parameter is either $\FalseB \in \Four$ or $\top \in \Four$, the result of the function \grant{} will be $\FalseB$, thereby denying the interaction as long as there is some policy that suggests so. This is inlined with the already-mentioned liberal approach, and that is also why the operator used for combining the policies coming from the source and from the target location in all the rules of Table \ref{tab:semantics_AspKBL_Reaction_Rules} is the operator $\oplus$. The definition of the \grant{} function in terms of four-valued Belnap Logic operand is:
$$
\grant{f} = f \leq_k \TrueB.
$$

Moreover, in the reaction rules of Table \ref{tab:semantics_AspKBL_Reaction_Rules}, the function \grant{} is called with an argument equal to the combination of the involved policies and also the intended action, whose result will be produced according to Table \ref{tab:semantics_AspKB_Policies_meaning}. The combination of the policies is done using the four-valued Belnap Logic operator $\oplus$, and this aims to produce a value in the set $\{ \top, \FalseB \}$ as long as some of the policies belongs to that very same set. This follows the same principle discussed in the previous paragraph, denying access as long as some policy suggests so. The intended action is also passed and the purpose is to check whether the involved policies actually trap that action (preformed by the function $check$ and $extract$, explained below), and that the applicability condition is met, otherwise the action should not be denied (this is achieved due to the $\bot$ returned by $\doublebracketleft . \doublebracketright$).

\begin{table*}[t]
$$
\begin{array}{l}

\db{[rec\ \IF\ cut: cond]}(l :: a\,.\, P)	\quad = \quad

\\	\qquad

\left(
	\begin{array}[c]{l}

	\hbox{case $check(\,extract(cut)\,;\,extract(l :: a\,.\, P))$ of}\

	\\	\qquad\qquad

		\begin{array}[t]{ll}
		\Fail:		&
			\bot
		\\
		\theta:	&
			\left\{
				\begin{array}{l@{\quad\textrm{if }}l}
				\dt{(rec\ \theta)}	&	\dt{cond\ \theta}
				\\
				\bot					&	\neg \dt{cond\ \theta}
				\end{array}
			\right.
		\end{array}

	\end{array}
\right)

\\[6ex]

\db{\neg pol}(N)	\quad = \quad
	\neg(\db{pol}(N))

\\

\db{pol_1\ \phi\ pol_2}(N)		\quad = \quad
	(\db{pol_1}(N))\ \phi\ (\db{pol_2}(N)), \
		(\phi \in \{\oplus,\otimes,\Rightarrow,>,\wedge,\vee\})

\\

\db{\TrueS}(N)				\quad = \quad
	\TrueB

\\

\db{\FalseS}(N)			\quad = \quad
	\FalseB

\end{array}
$$
\caption{Meaning of Policies in {\bf Pol} for {\bf AspectKBL}.}
\label{tab:semantics_AspKB_Policies_meaning}
\end{table*}

The function $check$ determines whether there is a substitution $\theta$ that can be performed in the $cut$ that matches the given argument, and it could be straightforwardly defined by induction on its arguments. The function $extract$ facilitates function $check$'s task by producing the list of literals that occur in a given syntactic construction in a way that, for instance, $extract(\ell :: \mathbf{out}(\ell_1^t,\cdots,\ell_n^t)@\ell'.X) = [\ell,\mathbf{out},\ell_1^t,\cdots,\ell_n^t,\ell',X]$, which is done by just pattern matching the components of the given parameter and then pushing them into a list.
\end{subsection}

\begin{subsection}{Example of network}
\label{subsec:examptinynetwork}
Let us discuss a tiny example to understand how it is written, and how the Semantics makes it evolve.

\begin{subsubsection}{Syntactic description}
\label{subsub:tinynetworkwithoutpolicies}
Assume that in a given Hospital we have a Health Care System where there is a centralised data base, named EHDB (for Electronic Health Data Base), with some information about some patients. In this case, let us assume there is one tuple (piece of data) regarding Alice, and that the tuple specifies a given Care Plan for her. Besides, there is another tuple regarding Bob, and it is related to some Private Notes some Doctor might have taken about him. This can be written in \textbf{AspectKBL} as follows:
$$
\begin{array}{l}
NetData =
\\
EHDB::^{pol_{EHDB}}<\texttt{Alice, CarePlan, alicetext}> \netpar
\\
EHDB::^{pol_{EHDB}}<\texttt{Bob, PrivateNotes, bobtext}>
\end{array}
$$
Note that, according to the syntax, although both tuples are in the same location, it must be written explicitly both times, since there is no Structural Congruence equivalence that allows to write it just once (as, for instance, the first equivalence in Table \ref{tab:semantics_AspKB_Structural_Congruence} for Processes). Note also, that each of the occurrences of the location has a security policy $pol_{EHDB}$ attached. Let us discuss later how to define the policy.

Assume now that there is also another location with information about the staff of the Hospital, which could be defined in the following way:
$$
\begin{array}{l}
NetRoles =
\\
ROLES::^{pol_{ROLE}}<\texttt{Doctor, Hansen}> \netpar
\\
ROLES::^{pol_{ROLE}}<\texttt{Nurse, Olsen}>
\end{array}
$$

Now, assume that both employees have some location, and there is a Process running on each of them. The Doctor Hansen might try to read patient Bob's information, and then leak it to the Nurse Olsen. On her side, Nurse Olsen might also try to read Bob's information directly. This could be defined as follows:
$$
\begin{array}{l}
NetHansen =
\\
Hansen::^{pol_{defaultDr}}
\\
\quad \textbf{read}(\texttt{Bob, PrivateNotes}, !content)@EHDB .
\\
\quad \textbf{out}(\texttt{Bob, PrivateNotes}, content)@Olsen .
\\
\quad \nil
\end{array}
$$
$$
\begin{array}{l}
NetOlsen =
\\
Olsen::^{\TrueS}
\\
\quad \textbf{read}(\texttt{Bob, PrivateNotes}, !content)@EHDB .
\\
\quad \nil
\end{array}
$$

Finally, the entire network could be defined using the previous definitions as follows:
\begin{equation}
\label{eq:netexample}
NetData \netpar NetRoles \netpar NetHansen \netpar NetOlsen
\end{equation}

Now, let us assume for the moment that there is no policy at all attached to any location (or, strictly speaking, there is a trivial policy, that always allows any interaction). According to what he have written, this could be directly obtained by defining:
$$
pol_{EHDB} = pol_{ROLE} = pol_{defaultDr} = \TrueS
$$
\end{subsubsection}

\begin{subsubsection}{Semantics evolving}
\label{subsubsec:Semantics_example}
Given the network for the Hospital we could obtain, for instance, the following possible path the network might follow, according to the Semantics:
$$
\begin{array}{ll}
& NetData \netpar NetRoles \netpar NetHansen \netpar NetOlsen
\\
\multicolumn{2}{l}{\rightarrow^{Hansen:\textbf{r}(\texttt{Bob, PrivateNotes, bobtext})@EHDB}}
\\
& NetData \netpar NetRoles \netpar NetOlsen \netpar
\\ & \quad Hansen::^{pol_{defaultDr}}\textbf{out}(\texttt{Bob, PrivateNotes, bobtext})@Olsen . \nil
\\
\multicolumn{2}{l}{\rightarrow^{Hansen:\textbf{o}(\texttt{Bob, PrivateNotes, bobtext})@Olsen}}
\\
& NetData \netpar NetRoles \netpar NetOlsen \netpar
\\ & \quad Hansen::^{pol_{defaultDr}} \nil \netpar
\\ & \quad Olsen::^{\TrueS}<\texttt{Bob, PrivateNotes, bobtext}>
\\
\multicolumn{2}{l}{\rightarrow^{Olsen:\textbf{r}(\texttt{Bob, PrivateNotes, bobtext})@Olsen}}
\\
& NetData \netpar NetRoles \netpar Olsen::^{\TrueS} \nil \netpar
\\ & \quad Hansen::^{pol_{defaultDr}} \nil \netpar
\\ & \quad Olsen::^{\TrueS}<\texttt{Bob, PrivateNotes, bobtext}>
\end{array}
$$

That is the path that is followed in case of both actions from the process in location $Hansen$ take place before the only action in location $Olsen$. Other two paths are possible if the interleaving of the actions is different:
$$
\begin{array}{ll}
& NetData \netpar NetRoles \netpar NetHansen \netpar NetOlsen
\\
\multicolumn{2}{l}{\rightarrow^{Hansen:\textbf{r}(\texttt{Bob, PrivateNotes, bobtext})@EHDB}}
\\
& NetData \netpar NetRoles \netpar NetOlsen \netpar
\\ & \quad Hansen::^{pol_{defaultDr}}\textbf{out}(\texttt{Bob, PrivateNotes, bobtext})@Olsen . \nil
\\
\multicolumn{2}{l}{\rightarrow^{Olsen:\textbf{r}(\texttt{Bob, PrivateNotes, bobtext})@Olsen}}
\\
& NetData \netpar NetRoles \netpar Olsen::^{\TrueS} \nil \netpar
\\ & \quad Hansen::^{pol_{defaultDr}}\textbf{out}(\texttt{Bob, PrivateNotes, bobtext})@Olsen . \nil
\\
\multicolumn{2}{l}{\rightarrow^{Hansen:\textbf{o}(\texttt{Bob, PrivateNotes, bobtext})@Olsen}}
\\
& NetData \netpar NetRoles \netpar Olsen::^{\TrueS} \nil \netpar
\\ & \quad Hansen::^{pol_{defaultDr}} \nil \netpar
\\ & \quad Olsen::^{\TrueS}<\texttt{Bob, PrivateNotes, bobtext}>
\end{array}
$$
and
$$
\begin{array}{ll}
& NetData \netpar NetRoles \netpar NetHansen \netpar NetOlsen
\\
\multicolumn{2}{l}{\rightarrow^{Olsen:\textbf{r}(\texttt{Bob, PrivateNotes, bobtext})@Olsen}}
\\
& NetData \netpar NetRoles \netpar NetHansen \netpar Olsen::^{\TrueS} \nil
\\
\multicolumn{2}{l}{\rightarrow^{Hansen:\textbf{r}(\texttt{Bob, PrivateNotes, bobtext})@EHDB}}
\\
& NetData \netpar NetRoles \netpar Olsen::^{\TrueS} \nil \netpar
\\ & \quad Hansen::^{pol_{defaultDr}}\textbf{out}(\texttt{Bob, PrivateNotes, bobtext})@Olsen . \nil
\\
\multicolumn{2}{l}{\rightarrow^{Hansen:\textbf{o}(\texttt{Bob, PrivateNotes, bobtext})@Olsen}}
\\
& NetData \netpar NetRoles \netpar Olsen::^{\TrueS} \nil \netpar
\\ & \quad Hansen::^{pol_{defaultDr}} \nil \netpar
\\ & \quad Olsen::^{\TrueS}<\texttt{Bob, PrivateNotes, bobtext}>
\end{array}
$$

These three paths are those occurring in the LTS induced by the Semantics of \textbf{AspectKBL}. In more complex networks, the LTS will be much more complex, consisting of many more paths, and much longer.
\end{subsubsection}

\begin{subsubsection}{With policies}
\label{subsub:tinynetworkwithpolicies}
Now, let us assume we do not want that Private Notes from any patient can be obtained by any Nurse, nor even given to them by any Doctor. Then, we could replace the trivial access control security policies that we put in the location for the following ones:

$$
\begin{array}{l}
pol_{EHDB} =
\Aspectbegin
\textbf{test}(\texttt{Doctor}, \#u)@ROLES
\\
\textif \#u::\textbf{read}(-,\texttt{PrivateNotes},-)@EHDB :
\\
\TrueS
\Aspectend[]
\end{array}
$$

$$
\begin{array}{l}
pol_{defaultDr} = 
\Aspectbegin
\textbf{test}(\texttt{Doctor}, \#target)@ROLES
\\
\textif \#u::\textbf{out}(-,\texttt{PrivateNotes},-)@\#target :
\\
\lnot (\#target = EHDB)
\Aspectend[]
\end{array}
$$

The first policy is the one attached to location $EHDB$, and it will avoid actions such as the one in the process of location $Olsen$. The second policy is the one attached to location $Hansen$ (and it could besides be attached to any other location that represents a Doctor), and it will avoid actions such as the second one in the process of location $Hansen$. Indeed, now the only possible path in the LTS induced by the Semantics would be:

$$
\begin{array}{ll}
& NetData \netpar NetRoles \netpar NetHansen \netpar NetOlsen
\\
\multicolumn{2}{l}{\rightarrow^{Hansen:\textbf{r}(\texttt{Bob, PrivateNotes, bobtext})@EHDB}}
\\
& NetData \netpar NetRoles \netpar NetOlsen \netpar
\\ & \quad Hansen::^{pol_{defaultDr}}\textbf{out}(\texttt{Bob, PrivateNotes, bobtext})@Olsen . \nil
\end{array}
$$

From that network state, there is no possible transition according to the Semantics. Indeed, if the policies were not there, there would be 2 possible actions that could execute, corresponding to the first two paths shown in Sub-subsection \ref{subsubsec:Semantics_example}. Ergo, the resulting LTS obtained when a network also consists of some security policies, could be interpreted as a kind of ``pruned'' LTS.

We could make the example a bit more complex by adding for instance a location $Administrator$ with a process that will change the position of employee Olsen, like this:
$$
\begin{array}{l}
NetUpgradeOlsen =
\\
Administrator::^{\TrueS}
\\
\quad \textbf{in}(\texttt{Nurse, Olsen})@ROLES .
\\
\quad \textbf{out}(\texttt{Doctor, Olsen})@ROLES .
\\
\quad \nil
\end{array}
$$
and by re-defining the entire network to be:
$$
NetData \netpar NetRoles \netpar NetHansen \netpar NetOlsen \netpar NetUpgradeOlsen
$$

Now, according to the specific interleaving taken, if that action is executed before the actions that would be denied in the previous case, then the policies would allow them, otherwise they will still be denied. For avoiding ``hackers'' to make such change in the $ROLES$ database, we could re-define the following policy:

$$
pol_{ROLES} = pol_{ROLES1} \oplus pol_{ROLES2}
$$
where
$$
\begin{array}{l}
pol_{ROLES1} = 
\Aspectbegin
\#u = Administrator
\\
\textif \#u::\textbf{in}(-,-)@ROLES :
\\
\TrueS
\Aspectend[]
\end{array}
$$
and
$$
\begin{array}{l}
pol_{ROLES2} = 
\Aspectbegin
\#u = Administrator
\\
\textif \#u::\textbf{out}(-,-)@ROLES :
\\
\TrueS
\Aspectend[]
\end{array}
$$
\end{subsubsection}

\end{subsection}

\begin{subsection}{Generic ``pruned'' LTS}
\label{subsec:prunedLTS}
In the previous Subsection we have seen a tiny example of how an LTS is induced by the Semantics, and we have noticed that if a network contains security policies then the induced LTS could be understood as a kind of ``pruned'' LTS, with shorter paths and perhaps fewer ones. Certainly, this pruned LTS is the actual one induced by the Semantics, so in the general case we could think that we are in a situation like the one in Figure \ref{fig:prunedLTS}:

\begin{figure}[ht]
  \centering
  \includegraphics[width=1\textwidth]{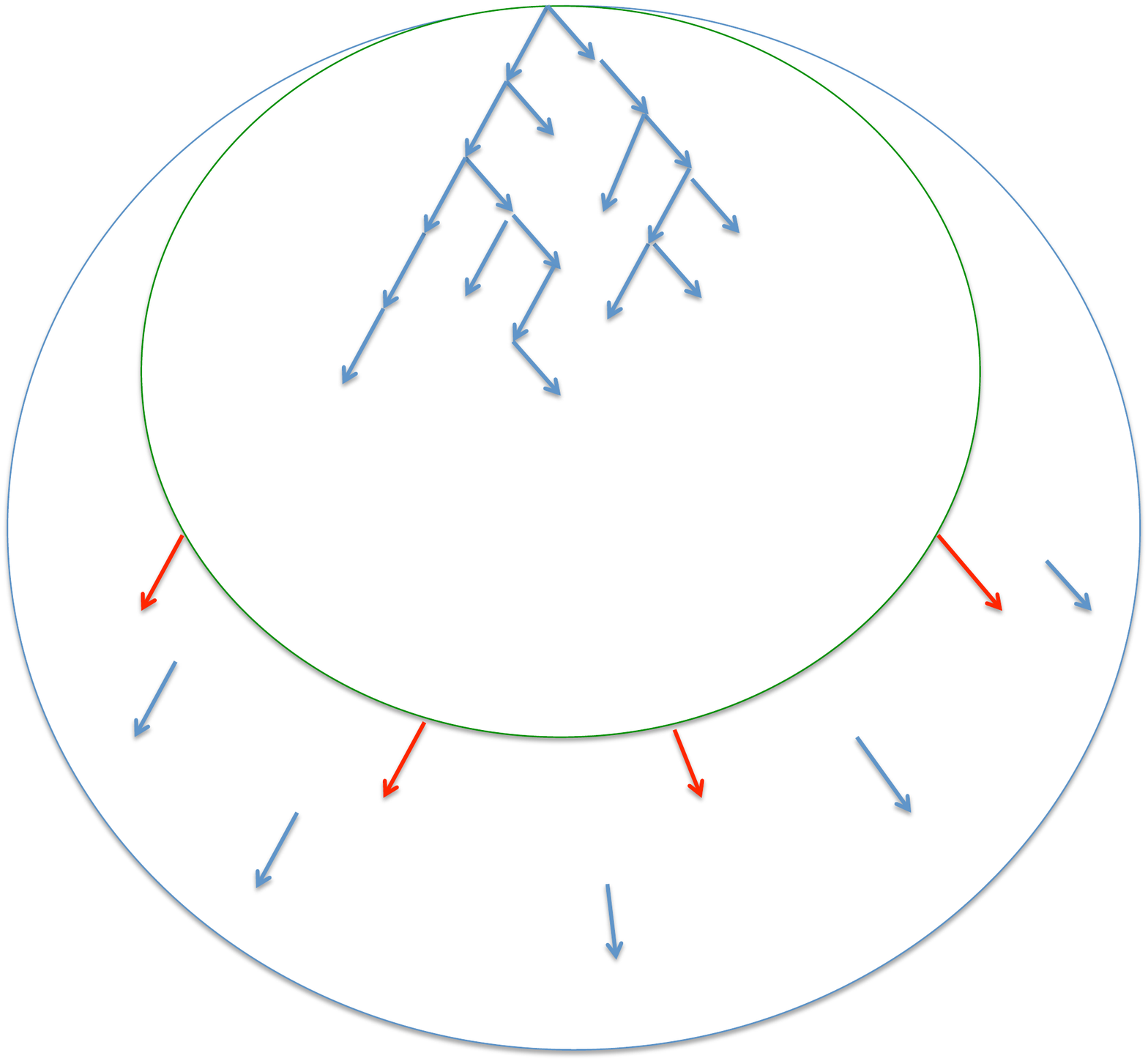}
  \caption{Generic pruned LTS.}
\label{fig:prunedLTS}
\end{figure}

Then, the LTS that would be induced if no security policies is different than the trivial one, is the outermost. However, the actual LTS we must be dealing with is the innermost, because the actions that might have happened in the ``border'' are those that are actually denied by some policy.
\end{subsection}

\end{section}

\begin{section}{A Logics for analysing Global Properties}
\label{sec:ACTLv}
Having defined the language for describing networks and localised security policies over them, we shall proceed on devising a technique for analysing the networks actually described using this language.

What we expect to have is a Logics for expressing the desired global security property of our network, and a way to check if the property is actually met by the network, considering the existing localised policies that we have attached. We approach the problem by defining a variant of the temporal logics ACTL \cite{ACTL} giving its syntax and semantics, and then we observe some properties useful for the later model checking of it.

\begin{subsection}{Defining the Logics}
\label{subsec:ACTLv_SyntaxSemantics}

We expect to describe useful desired global security properties, so let us assess what exactly might be a useful property to be described. As for global, what we need to establish is something that happens always, the system should always be secure in the sense of the property we might expect. As for security, what we need to establish is something that happens whenever some security threat might arise, the system should never actually fall into the threat thereby moving into an insecure state.

In a process calculi as the one we are dealing with, the interactions among locations are those that need to be monitored and controlled, and in particular when they happen some information may go from one location to another. Therefore, what we need to check and assess as possible threats are the movements of information that might not be desired. In such cases, we need to assure that the state reached after the interaction is secure.

Having said that, it shall be straightforward to realise that we need to trap all the possible interactions that are of our interest, and whenever they take place we need to check the states just before and just after the interaction, to see whether the interaction is leading to some insecure state. With this, the logic formula that naturally arises is the traditional $AG$, in our case annotated with some set of transitions, thereby converting our Logics in a variant of ACTL as already mentioned.

Moreover, the problem of properly characterising what security properties can indeed be enforced at runtime by access control methods has been dealt with by Schneider \cite{Schneider00}, and certainly they come up with the conclusion that \emph{safety} \cite{Alpern86recognizingsafety} properties are the answer. Then, as safety properties are those related with the $G$ modality of LTL
and the $AG$ one in CTL \cite{BaierKatoen08}, our assessment of the previous paragraph makes even more sense.

\begin{subsubsection}{Syntax}
The formal syntax of our Logic is given in Table \ref{tab:syntax_ACTLv}. We shall express an obligation (something we want the network to satisfy) as an $AG$ formula, meaning we want the property to be satisfied always, and in all possible paths the system might run into; this clearly enforces security. As a subscript to the formula, some set of transitions is to be given. The target of the transition must be a constant net, which is a restriction that will be later understood, but it is still useful as we will in general be willing to ensure security properties of databases or so. The idea is that some of the transitions in the running network might be trapped by the set of transitions given in this subscript, and in those cases the states related by the transition are to be analysed. The network states relating the transition are then analysed by checking the $Pred$ expressed in the formula, which can be a combination of smaller state predicates or the simplest ones comparing two values or testing the occurrence of some value in some location. For this last issue, the location that must be tested may be the one just before the transition (if no prime symbol is added) or the one just after the transition (if the prime symbol --$'$-- is added).

\newcommand{\tgt}{l}

\begin{table}[t]
$$
\begin{array}{l@{\quad}r@{\quad::=\quad}l}

Obl  \in  {\bf Obligations}	&	Obl
	&	AG_{\{labs\}}Pred
\\
labs  \in  {\bf Lab}			&	labs
	&	\Lc (\w_s) :  \locate{\textbf{c}({\veck{\ell^t}})}{\ell}\ (\w_t)
\\
\textbf{c} \in {\bf Cap}		&	\textbf{c}
	&	\textbf{o} \mid\ \textbf{i} \mid\ \textbf{r}
\\[.5ex]
Pred  \in  {\bf Predicates}	&	Pred
	&	\TrueS \mid \FalseS \mid \lnot Pred \mid Pred \lor Pred
\\
				&	\multicolumn{1}{l}{} 
	&	\mid Pred \land Pred \mid \forall x : Pred \mid \exists x : Pred \mid bp
\\
bp \in {\bf BasicPredicates}	&	bp
	&	\ell_a = \ell_b \mid \textbf{test}(\veck{\ell_a})@\ell_b \mid \textbf{test}'(\veck{\ell_a})@\ell_b
\mid \Gamma_1 \geq \Gamma_2
\\ \Gamma \in L & \Gamma & x \mid \gamma

\end{array}
$$
\caption{ACTLv Syntax -- How to express obligations.}
\label{tab:syntax_ACTLv}
\end{table}

\end{subsubsection}

\begin{subsubsection}{Semantics}


\begin{table}[t]
$$
\begin{array}{l}
N_0 \models_{Obl} AG_{\{cut\}} Pred
\\
\text{iff}
\\
\forall \text{ paths } N_0 \rightarrow^{*} N_i \rightarrow^{l_s:c(l)@l_t} N_{i+1} :
\\
\quad (\forall \theta : \, cut\,\theta = l_s:c(l)@l_t \,  \Rightarrow \, (N_i, N_{i+1}) \models_{Pr}^{\theta} Pred)
\end{array}
$$
\caption{ACTLv Semantics -- Satisfaction relation $\models_{Obl}$.}
\label{tab:semantics_ACTLv_ObligationsRightCoinductive}
\end{table}

The formal semantics of the Logics is divided into three satisfaction relations, one for each of the syntactic categories (\textbf{Obligations}, \textbf{Predicates} and \textbf{BasicPredicates}) defined in Table \ref{tab:syntax_ACTLv}. The first satisfaction relation gives semantics for the obligation formula and it is given in Table \ref{tab:semantics_ACTLv_ObligationsRightCoinductive}. It basically checks that in every path, when it is possible to substitute the $cut$ of the obligation thereby matching the label of the path's last transition, then the pair of nets that are connected by that transition satisfy the given predicate.

The satisfaction relation $\models_{Pr}$ is defined in Table \ref{tab:semantics_ACTLv_Predicates}. The rules should be straightforward, the only detail that shall be explained is in the three last rules. For the last one, if the predicate is a basic one, then the satisfaction relation $\models_{bp}$ is used. For the previous two, we have to make an extra substitution in the predicate for evaluating it, and the values that we might take are all those that occur in the involved nets, and for that purpose the auxiliary function $Loc$ is defined in the same table.

The satisfaction relation for basic predicates $\models_{bp}$ is given in Table \ref{tab:semantics_ACTLv_BasicPredicates}. The rules are straightforward, checking the equality in one case, and interpreting the \textbf{test} in the other two, distinguishing whether the \textbf{test} aims to check the net just before or just after the transition. In the three cases, the substitution is actually performed while checking the corresponding condition.

The way how the \textbf{test} is interpreted depends on the structure of the net, and its structural inductive definition is given in Table \ref{tab:semantics_ACTLv_testInterpretation}.


\begin{table}[t]
$$
\begin{array}{
	@{(N_1, N_2)}
	@{\quad\models_{Pr}^{\theta}\quad}
	l
	@{\quad\text{iff}\quad}
	l}
\TrueS			&	\TrueB
\\
\FalseS			&	\FalseB
\\
\lnot Pred			&	(N_1, N_2) \not\models_{Pr}^{\theta}Pred
\\
Pred_1 \lor Pred_2	&	(N_1, N_2) \models_{Pr}^{\theta}Pred_1	\lor
						(N_1, N_2) \models_{Pr}^{\theta}Pred_1
\\
Pred_1 \land Pred_2	&	(N_1, N_2) \models_{Pr}^{\theta}Pred_1	\land
							(N_1, N_2) \models_{Pr}^{\theta}Pred_1
\\
\forall x : Pred		&	\forall l \in 
Loc(N_1) \cup Loc(N_2)
:
						(N_1, N_2) \models_{Pr}^{\theta[l/x]}Pred
\\
\exists x : Pred		&	\exists l \in 
Loc(N_1) \cup Loc(N_2)
:
						(N_1, N_2) \models_{Pr}^{\theta[l/x]}Pred
\\
bp				&	(N_1, N_2) \models_{bp}^{\theta}bp
\end{array}
$$

$$
\begin{array}{r@{\quad=\quad}l}
Loc (l_1 :: \langle \veck{l_2} \rangle)	&	\{ l_1\} \cup Vec(\veck{l_2})
\\
Loc (l_s :: a(\veck{l_d}) @ l_t . P)		&	\{ l_s, l_t \} \cup Vec(\veck{l_d}) \cup Loc(P)
\\
Loc (N_1 \netpar N_2)		&	Loc(N_1) \cup Loc(N_2)
\end{array}
$$

$$
\begin{array}{r@{\quad=\quad}l}
Vec (\alpha, \veck{\alpha'})	&	\{ \alpha \} \cup Vec(\veck{\alpha'})
\\
Vec (\epsilon)				&	\emptyset
\end{array}
$$

\caption{ACTLv Semantics -- Satisfaction relation $\models_{Pr}$ and auxiliary functions $Loc$ and $Vec$.}
\label{tab:semantics_ACTLv_Predicates}
\end{table}

\begin{table}[t]
$$
\begin{array}{
	@{(N_1, N_2)}
	@{\quad\models_{bp}^{\theta}\quad}
	l
	@{\quad\text{iff}\quad}
	l}
\ell_a = \ell_b		&	(\ell_a \theta) = (\ell_b \theta)
\\
\textbf{test}(\ell_a)@\ell_b		&	\dt{ \textbf{test}(\ell_a \theta)@(\ell_b \theta), N_1 }
\\
\textbf{test}'(\ell_a)@\ell_b		&	\dt{ \textbf{test}(\ell_a \theta)@(\ell_b \theta), N_2 }
\\
\Gamma_1 \geq \Gamma_2		&	(\Gamma_1 \theta) \geq (\Gamma_1 \theta)
\end{array}
$$
\caption{ACTLv Semantics -- Satisfaction relation $\models_{bp}$.}
\label{tab:semantics_ACTLv_BasicPredicates}
\end{table}

\begin{table}[t]
$$
\begin{array}{
	@{[\!(\textbf{test}(}
	c
	@{)@l_2,\,}
	c
	@{\,)\!]\quad=\quad}
	l}
\veck{l_1}	&	N_1 \netpar N_2		&	\dt{ \textbf{test}(\veck{l_1})@l_2, N_1 } \lor
							\dt{ \textbf{test}(\veck{l_1})@l_2, N_2 }
\\
\veck{l_1}	&	\Localised{l}{pol}{P}		&	\FalseB
\\
\veck{l_1}	&	\Localised{l_3}{pol}{\val{\veck{\YZl_4}}}	&	(l_2 = l_3 \land \veck{l_1} = \veck{l_4})
\end{array}
$$
\caption{ACTLv Semantics -- Interpretation of \textbf{test}.}
\label{tab:semantics_ACTLv_testInterpretation}
\end{table}

\end{subsubsection}

\begin{subsubsection}{An example}
\label{subsubsec:ACTLvexample}
Continuing with our example of Section \ref{subsec:examptinynetwork}, we could decide to establish some global property we want the network to satisfy. We should be writing the property following the Syntax of our just-defined ACTLv Logics. The aim is that our property is met by the given network.

As discussed in Subsubsection \ref{subsub:tinynetworkwithpolicies}, we may aim at not allowing any Nurse to get access any Private Notes from any patient. We should then have some property that traps the transitions that could lead to such an ``insecure'' state. Clearly, in the LTS induced by the Semantics of \textbf{AspectKBL}, the transitions that could lead to such state are both if a Nurse reads the data directly or if they are given the data by some Doctor. Therefore, we will have to express two separate global properties, one capturing one case and another for the other case. Informally, we could express that by ``any \textbf{read} of Private Notes should be only done by a Doctor'' and ``any \textbf{out} of Private Notes should not be done to a Nurse's location''. Formally, we have to follow our syntactic conventions, thereby writing the following:
\begin{equation}
\label{eq:oblig1}
AG_{\{\$u:\textbf{r}(-,\texttt{PrivateNotes},-)@EHDB \}}\text{test}(\texttt{Doctor},\$u)@ROLES
\end{equation}
and
\begin{equation}
\label{eq:oblig2}
AG_{\{\$u:\textbf{o}(-,\texttt{PrivateNotes},-)@Olsen \}}\text{test}(\texttt{Doctor},Olsen)@ROLES
\end{equation}

Equation \ref{eq:oblig1} formalises the first informal property, but restricting to the database where the data might be (due to our syntactic restriction that the target location must be a constant). In Equation \ref{eq:oblig2} the restriction is a bit less practical, but still necessary due to our formal language. We need to specify which is the target location we are talking about, and in this case we would need to write one specific global property for each location that we suspect it might be a Nurse in our network. In our case, we do that with just location $Olsen$.

Later, we will be trying to automatically check if these (and other) properties hold in our network, and also in some other more complex ones.
\end{subsubsection}

\end{subsection}

\begin{subsection}{Structural Congruent nets produce same results}
\label{subsec:congruency}
One expected property of the Semantics defined in the previous Subsection is that if it is given two different \textbf{AspectKBL} networks, but which are actually structurally congruent, then the result should be the same. Indeed, otherwise it would mean that the result depends on how the network is described and not on which components it has, which in the end are the ones that make the network run. In this Subsection we establish a lemma about this issue.

\begin{lemma}
All the following rules hold:
$$
\begin{array}{lr}
\Inference
{N_1 \equiv N_2}
{N_1 \models_{Obl} Obl \iff N_2 \models_{Obl} Obl}
&
[StrC1]
\end{array}
$$
$$
\begin{array}{lr}
\Inference
{N_1 \equiv N_2 \land M_1 \equiv M_2}
{(N_1,M_1) \models_{Pr}^{\theta} Pred
	\iff
	(N_2,M_2) \models_{Pr}^{\theta} Pred}
&
[StrC2]
\end{array}
$$
$$
\begin{array}{lr}
\Inference
{N_1 \equiv N_2 \land M_1 \equiv M_2}
{(N_1,M_1) \models_{bp}^{\theta} bp
	\iff
	(N_2,M_2) \models_{bp}^{\theta} bp}
&
[StrC3]
\end{array}
$$
$$
\begin{array}{lr}
\Inference
{N_1 \equiv N_2}
{
\dt{ \textbf{test}(\veck{l_1})@l_2, N_1 }
	\iff
	\dt{ \textbf{test}(\veck{l_1})@l_2, N_2 }}
&
[StrC4]
\end{array}
$$
\end{lemma}

\begin{proof}{
We shall only prove the rule [StrC4], the other proofs follow the same fashion. For proving this rule, the definitions that we might use are those in Table \ref{tab:semantics_ACTLv_testInterpretation}, let's call them respectively $(a)$, $(b)$ and $(c)$ just for the purposes of this proof.

\newcommand{\test}[1]{\dt{ \textbf{test}(\veck{l_1})@l_2, #1 }}

Assuming $N_1 \equiv N_2$, we have to prove $\test{N_1} = \test{N_2}$. The proof is by induction on how $\equiv$ is obtained. We have four cases (three base- and one inductive-), according to Table \ref{tab:semantics_AspKB_Structural_Congruence}.

\textbf{Case} $N_1 = l::^{pol} P_1 \ppar P_2$ and $N_2 = l::^{pol}P_1 \netpar l::^{pol}P_2$
$$
\begin{array}{l}
\test{l::^{pol} P_1 \ppar P_2} \\
= [\text{by } (b)] \\
\FalseB \\
= [\text{boolean logic}] \\
\FalseB \lor \FalseB \\
= [\text{by } (b) \text{ twice}] \\
\test{l::^{pol}P_1} \lor \test{l::^{pol}P_2} \\
= [\text{by } (a)] \\
\test{l::^{pol}P_1 \netpar l::^{pol}P_2}
\end{array}
$$

\textbf{Case} $l ::^{pol}\ P$ and $l ::^{pol} P \netpar l ::^{pol} {\bf 0}$

\quad Analogous before since $\nil$ is a process.

\textbf{Case} $N_1 = l ::^{pol}\ * P$ and $l ::^{pol}\ P \ppar\ * P$

\quad Even more trivially since no application of $(a)$ is needed.

\textbf{Inductive case:} $N_1 = N \netpar N_a$ and $N_2 = N \netpar N_b$

\quad From the induction hypothesis we have that $\test{N_a} = \test{N_b}$.
$$
\begin{array}{l}
\test{N \netpar N_a} \\
= [\text{by } (a)] \\
\test{N} \lor \test{N_a} \\
= [\text{by Induction Hypthesis}] \\
\test{N} \lor \test{N_b} \\
= [\text{by } (a)] \\
\test{N \netpar N_b} \\
\end{array}
$$
}\end{proof}

Given that two congruent nets produce the same result when checking an ACTLv formula over them, we could rely on this to automatically check a given formula in any net that is structurally congruent to the one we are supposed to check. This gives the idea of single-representative for structurally congruent nets, and helps in choosing the one that fits better for an automatic checking. Indeed, when later we will be doing the automatic checking of the satisfaction formula, we will often be analysing some net equivalent to the given one.
\end{subsection}

\begin{subsection}{Interpretation of the Semantics over an LTS}
\label{subsec:LTSandDecidability}
As observed in Section \ref{sec:AspectKBL}, the Semantics of the \textbf{AspectKBL} language induces an LTS, and over such a structure it is possible to interpret the ACTLv Semantics from Subsection \ref{subsec:ACTLv_SyntaxSemantics}.

Now, if one could check that for all paths, it is always the case that, whenever some transition over that path is done and whose label matches the $cut$ of an ACTLv obligation then the predicate is satisfied in the reached state; then one would be verifying the formula.

However, since our \textbf{AspectKBL} language is Turing-complete\cite{KlaimTuringcomplete}, the paths might be infinite, and so
the breadth of the path tree. Although there are some results that show that for certain subclasses of LTS's the problem is decidable
\cite{Esparza99}\cite{BaierKatoen08}, we do not want to risk falling in a different subclass. Indeed, in \cite{Esparza99} it is also shown that, for instance, Branching Time Logics without $EG$ (and therefore $AF$) operator (a similar one to our ACTLv) is undecidable for Petri Nets, although it is decidable for some subclasses of Petri Nets, including BPP (Basic Parallel Processes).

\begin{subsubsection}{Decision taken to overcome the decidability issues}

For the moment, we shall concentrate in finite \textbf{AspectKBL} networks, and for achieving that we syntactically rule out any network that contains the replication operator $*$. This will of course limit our expressive power, but will at least give us the certainty that any network that is insecure in the sense of the ACTLv obligation we aim will be detected. Moreover, if some network is detected to be insecure, even without having any replication operator, then any extension of that network which includes some replication operators is also insecure.

\end{subsubsection}

\begin{subsubsection}{Interpreting our example}
In Subsubsection \ref{subsubsec:ACTLvexample} we have two examples of global properties that would be interesting to satisfy by a given network. If we take the network defined 
in Subsubsection \ref{subsub:tinynetworkwithoutpolicies}
and interpret the Semantics of Table \ref{tab:semantics_ACTLv_ObligationsRightCoinductive} over the LTS induced by the network, then we will see that the network does not actually satisfy the property.

However, if we take the network of Subsubsection \ref{subsub:tinynetworkwithpolicies}, that was slightly modified by adding some security policies to some of the locations, and then we again interpret the Semantics of Table \ref{tab:semantics_ACTLv_ObligationsRightCoinductive} over the LTS induced by the network, then we will see that the network does satisfy the property.

In the next Section, we will be assessing another way of checking this, without having to induce the whole LTS for deciding whether a given network satisfies a given property, as this would involve, for more complex networks (and in the general case), exponential time due to the state explosion problem.
\end{subsubsection}

\end{subsection}

\end{section}

\begin{section}{An alternative approach to model checking}
\label{sec:SmartModelChecking}
It is widely known that Model Checking suffers from the state explosion problem.
Several approaches have been taken for overcoming this limitation, some based on abstractions of the state space,
and some based on combination with other system assessing techniques such as static analysis.

In this Section, we propose an alternative approach, which works for our \textbf{AspectKBL} language, in which we avoid to explore the entire state space of the induced LTS while performing the model checking, by relying on some features of our language. We will be model checking \textbf{AspectKBL} networks not by inducing the entire LTS but by looking at the network definition, directly in the \textbf{AspectKBL} syntax.

We first informally introduce our approach (in Subsection \ref{subsec:smartmodelcheckinginformal}) and then formalise how we do it (in Subsection \ref{subsec:smartmodelcheckingformal}). Finally, an example is described (in Subsection \ref{subsec:modelcheckingexample}).

\begin{subsection}{Model checking by inspecting each single action}
\label{subsec:smartmodelcheckinginformal}
We know that the LTS induced by our \textbf{AspectKBL} network could be properly model checked following the traditional way. But we also know that this LTS depends directly on the specific network we have at hand. Moreover, the network is governed by some security policies attached to each location, and indeed those security policies never change during the runtime of the network, and this is actually ensured by the Semantics of Table \ref{tab:semantics_AspKBL_Reaction_Rules}. This is indeed a key point for the approach we shall take in this work. Certainly, since the policies never change during runtime, we may rely on them to be the ones that will in the end either allow or deny the actions to happen. This is indeed a key concept that helps us propose this alternative approach to model checking.

\begin{paragraph}{Over-approximation}
Our approach is to over-approximate the behaviour of the runtime network, with which it is possible to model check in a very fast way. We shall interpret the Semantics of the ACTLv formula statically, by looking directly to the \textbf{AspectKBL} network instead of to the LTS induced by it. The actions defined in each process are the basic concept, and this over-approximation allows us to check each and every action just once, instead of checking it in every path it may occur. Indeed, as whenever an action is matched by the $cut$ of the ACTLv obligation then the predicate of the obligation must be \TrueB, instead of checking the action considering the path in which it occurs we can just check the action by itself, relying that it will certainly be governed by some security policies, which will not change during runtime so we could know them, and the possible decisions of them, beforehand. Therefore, if the Belnap recommendation of the security policies that govern a given action implies the intended predicate, we can safely certify the security of the given action.

Moreover, if one recalls how an LTS is induced from a given \textbf{AspectKBL} network, it is straightforward that the security policies are directly involved on how the Semantics influence the induction of every path. From there, one could interpret the Semantics of ACTLv after every path has been obtained. Therefore, we could think on an ``order'' on how the operations take part, or how are they stepwise. Indeed, the security policies are considered first, and then the global property should be interpreted over the resulting LTS. Since our approach will consider each action just once, and since the policies that govern such action are fixed at runtime, we could rely on them to know whether the action will be allowed or not. The procedure then changes slightly. Instead of inducing the LTS already considering the aspects, and then having just to check the relevant transitions against the global property, here we take each action exactly once, although in some cases it might be granted and in some it may not, and although the actual values bound to the occurring variables might change dramatically. Indeed, what we have now is a more general (and over-approximate) way of checking each action. However, the actions are still governed by some aspects of security policies, which will be those that in the end will either grant or deny the action.

Actually, and since the action might contain a variable as target, we may need to ground that variable first, so we could count on a given set of policies (the ones coming from the source of the action, which is known because it is taken from the description of the network, and the ones coming from the target, which after grounding the variable is also known).
\end{paragraph}

\begin{paragraph}{Determining involved policies}
The process of grounding the target, if it is indeed a variable, is done by inverting the order in which the operations mentioned in the previous paragraph take part. Taking into account that in the $cut$ of an ACTLv Obligation the target must be a constant (restricted by the Syntax of Table \ref{tab:syntax_ACTLv}), if we first check whether the action might be trapped by the $cut$ then we could proceed by considering the action, otherwise we can safely (and trivially) certify the security of the action, since it is not relevant for our purposes. If the action is relevant, and since the target in the $cut$ is a constant, we may safely assume that the action will be relevant only in case the target variable is ground to the specific constant value occurring in the $cut$. To illustrate, assume the action is $l_1::\textbf{read}(x)@y$ and the $cut$ of the global property is $\$x::\textbf{read}(\$y)@l_2$. The action could occur in the induced LTS with the variables $x$ and $y$ bound to various different values, for instance $l_1::\textbf{read}(l_1)@l_1$, $l_1::\textbf{read}(l_2)@l_3$, $l_1::\textbf{read}(l_2)@l_1$, $l_1::\textbf{read}(l_3)@l_3$, etc. However, only those occurrences that have the value of variable $y$ equal to $l_2$ are relevant for the global property, for instance $l_1::\textbf{read}(l_1)@l_2$, $l_1::\textbf{read}(l_2)@l_2$, $l_1::\textbf{read}(l_3)@l_2$, etc.

Indeed, some other variables occurring either in the action or in the $cut$ may also be ground during this matching, and only then we start assessing whether the resulting (possibly completely grounded) action might be allowed or not by the (fixed) security policies.

In future work, we may relax the syntactic restriction that the target in the $cut$ of the obligation must be a constant, but the some Static Analysis might be needed in order to obtain a set of possible values for the occurring variables, in order to determine the possible involved policies.
\end{paragraph}

\begin{paragraph}{Using policies}
Now, if we can certainly assure that the policies will not allow the action, again we can certify that the action is secure (other possible values for the variables are not relevant for the global property, and the ones found with this grounding will never be used for an actual transition). However, if the action may possibly be allowed by the security policies, our task is slightly different, although we can still do some work. Indeed, if the policies may allow the action to take place, then their recommendations $rec$ must be \TrueB, but then we can rely on them to be sure that some properties hold whenever the action might be allowed. The properties that hold are exactly those implied by the constraints established in the $rec$ of the involved security policies. If they imply the property established in the predicate $Pred$ of the obligation, then we can also certify the action we are analysing.

The procedure basically involves taking one by one the aspects and finding a substitution that can be applied to both the cut of the aspect and the action in order to unify them. If such substitution does not exist, we ignore the aspect, as it will give bottom (thereby granting the action when mapping into 2-valued logic). Otherwise, if there is such substitution, we need to apply it to the condition of the aspect to see if it might be considered, otherwise (if it will never be considered given such substitution, it is also bottom and thereby grant) we also ignore it. After knowing the aspect might be considered, we take the same substitution and apply it to the recommendation, and we have to count with the resulting constraint and proceed with the next aspects (in each finding its own substitution), until we at the end combine all the constraints obtained using 4-valued logic and then map to 2-valued for deciding whether the action might be granted. If it cannot be granted then we can certify the action, otherwise we need to prove that the combination of the constraints implies the predicate of the global property.
\end{paragraph}

\begin{paragraph}{Summing up}
If each and every action is certified following that procedure, we can certify the whole network. Otherwise, if we cannot safely certify some of the actions, then we cannot conclude anything about the entire network.

This procedure does not need any LTS to be induced, nor the hypothetical one assuming no security policy is different than the trivial one, and neither the pruned one according to the security policies (recall Subsection \ref{subsec:prunedLTS}). Certainly, since every action is considered and their variables are ground just by matching with the $cut$ of the obligation, thereby detecting possible relevance of the action for the desired global property, then we cannot even know if those values for the variables can be indeed possible in the induced LTS. Moreover, this could be the source of an imprecise over-approximation, since we may finally not certify the network due to some possible action we cannot certify, but the action might actually never occur in the pruned LTS. It is then a key subject of study to determine the circumstances in which we could rely on the precision of our procedure. On the other side, if the over-approximation is indeed safe, we can certainly rely on a network that has been certified. For achieving this, the correctness of our approach shall be proven.
\end{paragraph}
\end{subsection}

\begin{subsection}{The algorithm for model checking}
\label{subsec:smartmodelcheckingformal}
The main algorithm for performing our alternative approach to model checking basically does what is explained in the previous Subsection. Here we shall see more formally how it achieves that, and show some auxiliary parts of it.

\begin{paragraph}{Unifying}
First, recall that we have to match in several occasions with the action we are analysing in each step. For performing the matching, some substitution must be found, that unifies the action with the given entity that is being matched. Whenever we are given one entity, which could be the $cut$ of an obligation or the $cut$ of an aspect, we must try to match it with the action we are analysing, and the function for making the unification is defined as follows:

$$
\begin{array}{l}
unify\ l_1\ l_2 = \text{if } l_1 = l_2 \text{ then } id \text{ else } fail \\
unify\ l_1\ v_2 = [ v_2 \mapsto l_1 ] \\
unify\ v_1\ l_2 = [ v_1 \mapsto l_2 ] \\
unify\ v_1\ v_2 = [ v_1 \mapsto v_2 ] \\
unify\ '-'\ l_2 = id \\
unify\ '-'\ v_2 = id
\end{array}
$$

One should notice that in the fourth line the direction of the mapping is always from the first to the second parameter. This must be as it is because the actual action being analysed will always contribute to the second parameter, and since we need to find possible relevance of the action to the $cut$ being considered, we have to map variable from the $cut$ into the action. Besides, the last 2 lines capture the case where in the $cut$ it is ignored the value of some given parameter, and in such cases the action might of course be trapped no matter which is the component on it. Finally, a $fail$, not only here but in later functions, means that there was not possible to find a matching, and in our case it will mean that by no means the action might be relevant for the $cut$ we are taking.

That unification function is only for single literals, but while trying to match $cut$ with actions, there are indeed several literal occurring, and moreover the number is unknown, and even unbound, due to the ones that can occur as parameters of the action. Therefore, we need an extended function that allows capturing those cases:

$$
\begin{array}{l}
unifylist\ nil\ nil = id \\
unifylist\ (x:xs)\ (y:ys) = \theta_1 . \theta_2 \\
\quad	where \\
\qquad	\theta_1 = unify\ x\ y \\
\qquad	\theta_2 = unifylist\ (xs\ \theta_1) (ys\ \theta_1) \\
unifylist\ nil\ (y:ys) = fail \\
unifylist\ (x:xs)\ nil = fail
\end{array}
$$

It should be noticed the order in which the substitution pairs are put into the substitution sequence, in particular in the second line. This is done as it is because it directly depends on how it is later used for performing the substitution, with the sequence found. Indeed, given a substitution consisting of a sequence of several pairs, we shall start applying from the \emph{beginning} of the sequence, and then continue by applying the rest of the substitution pairs to the obtained entity, and so on.
\end{paragraph}

\begin{paragraph}{Matching whole action}
Using the unification function defined in the previous paragraph, we need to find an entire substitution that allows the action being analysed to be matched, or trapped, by the specific $cut$ considered right now, namely a substitution $\theta$ such that $cut\ \theta$ = $act\ \theta$, assuming the action is $act$.

This task is performed by a function $findsubs$, that takes the $cut$ and the action, and it stepwise performs the unifications, applying the substitution parts already found to later literals that are to be unified. The definition is the following:

$$
\begin{array}{l}
findsubs\ cut\ action = \theta_1 . \theta_2 . \theta_3\ \\
\quad	where \\
\qquad	\theta_1 = unify\ locsrc1\ locsrc2 \\
\qquad	\theta_2 = unifylist\ (params1\ \theta_1)\ (params2\ \theta_1) \\
\qquad	\theta_3 = unify\ ((loctgt1\ \theta_1) \theta_2)\ ((loctgt2\ \theta_1) \theta_2)\\
\quad and\ assuming \\
\qquad	cut = locsrc1\ params1\ loctgt1 \\
\qquad	action = locsrc2\ params2\ loctgt2
\end{array}
$$

Again, notice that the order, in the first line, is set in such way that then the entire substitution is to be applied starting from the beginning of the sequence. Moreover, the first line will give $fail$ as long as at least one of the components gives $fail$. It is also worth noticing that finding a substitution in such way is only valid provided both the $cut$ and the action are defined to be the same type of operation, the same \emph{capability}, namely \textbf{read}, \textbf{in} or \textbf{out}.
\end{paragraph}

\begin{subsubsection}{Applying algorithm}
As discussed in Subsection \ref{subsec:smartmodelcheckinginformal}, the algorithm for model checking will take each action in its turn, and it will check if it is relevant for the $cut$ of the obligation. This is done by trying to match the $cut$ with the action. If it is possible, then it will use the substitution found in order to see if under that condition the action might be granted by the policies, otherwise it is trivially certified. To detect if the action might be granted by the policies, the combination of them must be considered, and within each of them every single aspect must be considered in its turn. For each aspect, first it is necessary to see if it is relevant for the given action, by finding a substitution with its $cut$, and in case it is, then the condition $cond$ and later the recommendation $rec$ must be substituted and evaluated. Finally, if the $rec$ might be $\TrueB$ then we could count on that as one specific constraint that will hold whenever the action can be executed, and check if under the set of constraints found the predicate of the obligation will always be satisfied.

That is formalised under the following algorithm, which is given by stepwise dividing its parts, in different levels of abstraction. Broadly, the algorithm consist of three parts, the first one in charge of taking each and every action and analysing it in an isolated way. The second part is the one in charge of analysing how a policy influences the given action. The third part is the one in charge of evaluating the constraints obtained in order to see whether they imply the expected predicate.

\begin{paragraph}{Isolating action}
This first part isolates the action, and it makes the alternative calls to other functions to determine in which cases the action might be certified. It consists of three subparts: the one in charge of the entire algorithm, the one in charge of deciding whether the specific action is relevant for the global property, and the one in charge of checking if a relevant action complies with the obligation.

The main algorithm just splits the entire network into single actions, for checking them separately. If each and every actions satisfies the predicate then the entire network does so, otherwise it may not:
\begin{verbatim}
[MAIN algorithm: Checks if Network N satisfies Obligation Obl]

If each action A from Network N satisfies Obl
Then Return True
    (obligation is satisfied)
Else Return False
    (obligation may not be satisfied)
\end{verbatim}
In capital letters it is referenced to some other parts of the algorithm, or to some function already formally defined. Besides, the sentences between parenthesis are comments of what the output of the pseudo code means.

Once each action is isolated, the certification of it must be done by first checking its relevance to the given global property, by trying to match it using some possible substitution. If no substitution can be found, then the action trivially satisfies the obligation, otherwise its compliance has to be assessed:
\begin{verbatim}
[Split for actions: Checks if Action A satisfies Obligation Obl]

If FINDSUBS (cut) (A) can find a substitution
Then CHECKCOMPLIANCE of A under the substitution found
Else Return True
    (obligation is satisfied)
(where cut is the cut of the Obligation Obl)
\end{verbatim}

For checking the compliance of the action, first it has to be checked whether the action might indeed be executed at all, because otherwise the action is again trivially certified. In case the action might indeed be executed, then the constraints raised by the policies must be checked:
\begin{verbatim}
[CHECKCOMPLIANCE algorithm: Checks if Action A
satisfies Obligation Obl, under the given substitution Theta0]

If both policies coming from the source and from the target
of action (A Theta0) MIGHTGRANT action (A Theta0)
Then CHECKCONSTRAINTS given by the policies
Else Return True
    (obligation is satisfied)
\end{verbatim}
These two auxiliary procedures are the points of the other parts of the algorithm.
\end{paragraph}

\begin{paragraph}{Policy influence}
This part is the one in charge of deciding whether the involved policies might grant the action, and collecting some constraints in case it is, to recall the conditions under which the action might indeed be granted.

For a complex set of policies, each of them must be checked individually, and then their results must be taken and combined, according to the Belnap operators that are used to combine them in the locations where they are attached:
\begin{verbatim}
[MIGHTGRANT algorithm: Checks if there might be some
chance that action A is allowed by the policy Pol]

If Belnap Combination of every SINGLEPOLICY in Pol
might grant action A
Then Return True
    (might grant -> don't know yet about obligation)
Else Return False
    (never grants -> obligation is satisfied)
\end{verbatim}

Given a single involved policy, for checking its decision for a specific action, the first step is to find whether its cut can match the action, by finding some possible substitution. If no substitution is found, then the policy is not actually considered, allowing the action by default. In case there is some substitution, the applicability condition of the policy must be assessed, given the substitution found:
\begin{verbatim}
[SINGLEPOLICY algorithm: Checks if there might be some
chance that action A is allowed by the single aspect Asp]

If FINDSUBS (cut) (A) can find a substitution
Then CHECKCONDITION under the substitution found
    (don't know yet if it may grant)
Else Return True
    (might grant -> don't know yet about obligation)
(where cut is the cut of the Aspect Asp)
\end{verbatim}

If the applicability condition says the policy must be applied, then its recommendation will give the final decision of it:
\begin{verbatim}
[CHECKCONDITION algorithm: Checks if the aspect Asp
might be applied, under the given substitution Theta1]

If (cond Theta1) might be True
Then CHECKRECOMMENDATION of aspect Asp
    (don't know yet if it may grant)
Else Return True
    (might grant -> don't know yet about obligation)
\end{verbatim}

While checking the recommendation of a single policy, if it will never grant the action then the action can be trivially certified, otherwise the conditions under which the action might be granted have to be collected for later assessment:
\begin{verbatim}
[CHECKRECOMMENDATION algorithm: Checks if the aspect Asp
might grant action A, under the given substitution Theta1. If it
might, then enqueue new constraint]

If (rec Theta1) might be True
Then ENQUEUECONSTRAINT (rec Theta1)
    (might grant -> don't know yet about obligation)
Else Return True
    (never grants, always denies -> obligation is satisfied)
(where rec is the recommendation of aspect Asp)
(ENQUEUECONSTRAINT is not defined, but used within CHECKCONSTRAINTS)
\end{verbatim}
\end{paragraph}

\begin{paragraph}{Solving the constraints}
After determining which single policies would allow the action to occur, and under which conditions, we need to solve the set of constraints given by those conditions in order to determine if under those conditions the predicate of the obligation that we are analysing is satisfied. If that is the case, then we can certify the action, as it will always be the case that if the action is allowed, it is due to the recommendation of the involved policies, and because they form a set of conditions that imply the predicate, then the action is indeed secure:
\begin{verbatim}
[CHECKCONSTRAINTS algorithm: Checks if the combination
of recommendation from the policies imply the predicate Pred.
under the given substitution Theta0]

If set of Enqueued constraints implies (Pred Theta0)
Then True
    (implication holds -> obligation is satisfied)
Else False
    (implication does not hold -> obligation may not be satisfied)
\end{verbatim}
\end{paragraph}

\end{subsubsection}
\end{subsection}

\begin{subsection}{Model Checking our example}
\label{subsec:modelcheckingexample}
Let us illustrate the model checking procedure with our running example. The network is the one from Equation \ref{eq:netexample}, in Section \ref{sec:AspectKBL}, and let us take the Obligation from Equation Equation \ref{eq:oblig1}, in Section \ref{sec:ACTLv}.

The list of all the single actions described in the network is the following:
$$
\begin{array}{l}
Hansen::\textbf{read}(\texttt{Bob, PrivateNotes}, !content)@EHDB
\\
Hansen::\textbf{out}(\texttt{Bob, PrivateNotes}, content)@Olsen
\\
Olsen::\textbf{read}(\texttt{Bob, PrivateNotes}, !content)@EHDB
\end{array}
$$
The first two coming from the sub-network named $NetHansen$ and the last one from the one named $NetOlsen$.

Our alternative model checking suggests that we are to take each and every action from that list and check it separately. Let us start then with the first one, $Hansen::\textbf{read}(\texttt{Bob, PrivateNotes}, !content)@EHDB$, and proceed with the checking. The first task to do now is to check whether the action can be matched by the $cut$ of the Obligation, which is $\$u:\textbf{r}(-,\texttt{PrivateNotes},-)@EHDB$. The unifying and matching functions from Subsection \ref{subsec:smartmodelcheckingformal} help us to find the substitution $\theta = [\$u \mapsto Hansen]$. Certainly, $findsubs\ cut\ action = \theta_1 . \theta_2 . \theta_3$, where
$$
cut = \$u\ (-,\texttt{PrivateNotes},-)\ EHDB
$$
and
$$
action = Hansen\ (\texttt{Bob, PrivateNotes}, !content)\ EHDB,
$$
and where $\theta_1$, $\theta_2$ and $\theta_3$ are obtained by:
$$
\begin{array}{l}
\theta_1 =
\\
unify\ \$u\ Hansen = [\$u \mapsto Hansen]
\\[2ex]
\theta_2 =
\\
unifylist\ ((-,\texttt{PrivateNotes},-)\ \theta1)\ ((\texttt{Bob, PrivateNotes}, !content)\ \theta_1) =
\\
unifylist\ (-,\texttt{PrivateNotes},-)\ (\texttt{Bob, PrivateNotes}, !content) = id . id . id
\\[2ex]
\theta_3 =
\\
unify\ EHDB\ EHDB = id
\end{array}
$$

The second task to do is to determine the involved policies, and they are indeed $pol_{defaultDr}$ and $pol_{EHDB}$, the former coming from the source of the action while the latter from the target. We need to check whether they both might grant the action to occur. The action that must be checked is the one after applying the substitution just found, which in this case gives the very same action.

For each of the policies, then, the following is applied. First, it is checked whether the $cut$ of the policy can match the action. Indeed, taking, let us say, $pol_{EHDB}$, we find the substitutition $[\#u \mapsto Hansen]$. Second, it is checked whether the policy condition $cond$ might apply. Indeed, in our case it does as it is the constant $\TrueS$. Third, it is checked whether the recommendation $rec$ might evaluate to $\TrueB$. Indeed, applying the substitution we get
$$
(\textbf{test}(\texttt{Doctor}, \#u)@ROLES) [\#u \mapsto Hansen]
$$
which ends up being $\textbf{test}(\texttt{Doctor}, Hansen)@ROLES$, evaluating to $\TrueB$ as the tuple occurs in $NetRoles$.

That constraint ($\textbf{test}(\texttt{Doctor}, Hansen)@ROLES$) would then be enqueued for later solving when evaluating the predicate of the obligation. In our case, this reduces to a single checking as we use the same substitution obtained before ($\theta = [\$u \mapsto Hansen]$), and the predicate is then the same as the constraint. If we could prove the predicate given the constraints taken from the entire set of policies, then we could certify that particular action.

To certify the entire network we have to proceed the same with each and every action. Assume then, that we had taken the third action, Olsen::\textbf{read}(\texttt{Bob, PrivateNotes}, !content)@EHDB, and proceed in the same way as before. The substitution found is then $\theta = [\$u \mapsto Olsen]$. The involved policy is just $pol_{EHDB}$, as location $Olsen$ has the trivial policy $\TrueS$. Again, while checking if the policy applies, we end up finding that with the substitution $[\#u \mapsto Olsen]$ it might be applied, but it will not allow the action, as the substituted recommendation $(\textbf{test}(\texttt{Doctor}, \#u)@ROLES) [\#u \mapsto Olsen]$ ends up being $\textbf{test}(\texttt{Doctor}, Olsen)@ROLES$, which does not evaluate to $\TrueB$. Therefore, we can also certify the action, in this case a little more trivially (it will not be executed). If we had not have the policy, we could not have certified the action, as we could not have said that the action will not be executed, and then we should have proven the predicate of the obligation, without having any enqueued constraint (as before).
\end{subsection}

\end{section}

\begin{section}{Examples}
\label{sec:examples}
In this Section we show 2 more elaborated examples, in a case-study-like way. The first one is an extension of the examples in Sections \ref{subsec:examptinynetwork}, \ref{subsubsec:ACTLvexample} and \ref{subsec:modelcheckingexample}, with which we have been illustrating our developments. The second one comes from the Web community, and it is an abstraction of a known problem that has happened in actual systems, and we show that our formal framework is able to capture it and detect the insecurity. We aim at being able to detect also insecurities that are not known yet, by putting more detail to the example (or formalising some other from the same domain) and playing it out with our framework.

\begin{subsection}{A health care example}
In this example, based on \cite{AspectEHR} but with several changes and addings, we assume there is a Health Care Data Base (location $EHDB$) and some role-based access control rules depending on the roles of the process locations. There is a location $ROLES$ that defines this, and then some Doctors and Nurses. The general aim is that nurses are not allowed to get access to doctors' private notes.

The complete description of the system is the following:
$$
\begin{array}{l}
\textbf{Location } EHDB:
\\
\text{Tuples:}
\begin{array}{|l|l|}
\hline
patiendid & \text{the id of the patient}
\\ \hline
recordtype & \text{either \texttt{PrivateNotes} or \texttt{MedicalRecord}}
\\ \hline
author & \text{the Doctor id who created it}
\\ \hline
creationtime & \text{either \texttt{Past} or \texttt{Recent}} 
\\ \hline
data & \text{the data stored in the record}
\\ \hline
\end{array}
\\
\text{Policy: }(pol_{EHDB})
\\
\Aspectbegin
\textbf{test}(\texttt{Doctor}, \#u)@ROLES
\\
\textif \#u::\textbf{read}(-,\#type,-,-,-)@EHDB :
\\
\#type = \texttt{PrivateNotes}
\Aspectend[]
\\[5ex]
\textbf{Location } ROLES:
\\
\text{Tuples:}
\begin{array}{|l|l|}
\hline
roletype & \text{either \texttt{Doctor} or \texttt{Nurse}}
\\ \hline
id & \text{the id of the employee}
\\ \hline
\end{array}
\\
\text{Policy: $pol_{ROLE}$ Not explicited here to simplify the example, there should be policies}
\\
\text{to allow creation or deletion of records only by the Administrator}
\end{array}
$$

$$
\begin{array}{l}
\textbf{For each DOCTOR location in the system: (for example } Smith)
\\
\text{Policy:} (pol_{defaultDr})
\\
\Aspectbegin
\textbf{test}(\texttt{Doctor}, \#target)@ROLES
\\
\textif \#u::\textbf{out}(-,\texttt{PrivateNotes},-,-,-)@\#target :
\\
\lnot \#target = EHDB
\Aspectend[]
\end{array}
$$

The global property we want to enforce is that we don't want private notes in any nurse's location. A bit closer to the Logics we have, we can (still informally) express the property by: any \textbf{read} of PrivateNotes should be only done by a Doctor and any \textbf{out} of PrivateNotes should not be done to any Nurse's location.
Formally, that can be written:
\begin{equation}
\label{eq:prop1}
AG_{\{\$u:\textbf{r}(-,\texttt{PrivateNotes},-,-,-)@EHDB \}}\text{test}(\texttt{Doctor},\$u)@ROLES
\end{equation}
and the following three for the second informal property:
\begin{equation}
\label{eq:prop2a}
AG_{\{\$u:\textbf{o}(-,\texttt{PrivateNotes},-,-,-)@Olsen \}}\text{test}(\texttt{Doctor},Olsen)@ROLES
\end{equation}
and
\begin{equation}
\label{eq:prop2b}
AG_{\{\$u:\textbf{o}(-,\texttt{PrivateNotes},-,-,-)@Smith \}}\text{test}(\texttt{Doctor},Smith)@ROLES
\end{equation}
and
\begin{equation}
\label{eq:prop2c}
AG_{\{\$u:\textbf{o}(-,\texttt{PrivateNotes},-,-,-)@Hansen \}}\text{test}(\texttt{Doctor},Hansen)@ROLES
\end{equation}
due to the syntactic restriction that we still have.

Now assume we have the following data already stored:
$$
\begin{array}{l}
EHDB::^{pol_{EHDB}}<\texttt{Alice, MedicalRecord, Hansen, Past, alicetext}> \netpar
\\
EHDB::^{pol_{EHDB}}<\texttt{Bob, PrivateNotes, Smith, Recent, bobtext}> \netpar
\\
ROLES::^{pol_{ROLE}}<\texttt{Doctor, Hansen}> \netpar
\\
ROLES::^{pol_{ROLE}}<\texttt{Doctor, Smith}> \netpar
\\
ROLES::^{pol_{ROLE}}<\texttt{Nurse, Olsen}>
\end{array}
$$
and let's call that networt $Data$.

Now assume three separate examples of process locations running in parallel with it. Therefore we will have:

\noindent \textbf{Example1:}
$$
Data \netpar HansenGood \netpar OlsenGood
$$
\textbf{Example2:}
$$
Data \netpar HansenBad \netpar OlsenGood
$$
\textbf{Example3:}
$$
Data \netpar HansenGood \netpar OlsenBad
$$

Where we can define the process locations as follows:
$$
\begin{array}{l}
HansenGood =
\\
Hansen::^{pol_{defaultDr}}
\\
\quad \textbf{read}(\texttt{Alice, MedicalRecord, Hansen, Past}, !content)@EHDB .
\\
\quad \textbf{out}(\texttt{Alice, MedicalRecord, Hansen, Past}, content)@Olsen .
\\
\quad \textbf{read}(\texttt{Bob, PrivateNotes, Smith, Recent}, !content)@EHDB .
\\
\quad \textbf{out}(\texttt{Bob, PrivateNotes, Smith, Recent}, content)@Hansen .
\\
\quad \nil
\end{array}
$$
$$
\begin{array}{l}
HansenBad =
\\
Hansen::^{pol_{defaultDr}}
\\
\quad \textbf{read}(\texttt{Bob, PrivateNotes, Smith, Recent}, !content)@EHDB .
\\
\quad \textbf{out}(\texttt{Bob, PrivateNotes, Smith, Recent}, content)@Olsen .
\\
\quad \nil
\end{array}
$$
$$
\begin{array}{l}
OlsenGood =
\\
Olsen::^{\TrueS}
\\
\quad \textbf{read}(\texttt{Alice, MedicalRecord, Hansen, Past}, !content)@EHDB .
\\
\quad \nil
\end{array}
$$
$$
\begin{array}{l}
OlsenBad =
\\
Olsen::^{\TrueS}
\\
\quad \textbf{read}(\texttt{Bob, PrivateNotes, Smith, Recent}, !content)@EHDB .
\\
\quad \nil
\end{array}
$$

\begin{subsubsection}{Model Checking and results}
Now, a model checking done in the entire LTS induced by the Semantics for Example1, Example2 and Example3 will indeed prove that the three of them satisfy both properties of Equations \ref{eq:prop1} and \ref{eq:prop2a} (and of course \ref{eq:prop2b} and \ref{eq:prop2c}). Moreover, the alternative way of model checking defined in Section \ref{sec:SmartModelChecking} is also able to prove that.

Furthermore, if we take out all the policies and just attach \TrueS\ to each location (as it is done, for instance, in $Olsen$ location), then only Example1 will satisfy both properties. Example2 will fail to satisfy property of Equation \ref{eq:prop2a} and Example3 will do so with property of Equation \ref{eq:prop1}. Moreover, our alternative way of model checking is also able to detect that these properties are not satisfied in each case.
\end{subsubsection}

\end{subsection}

\begin{subsection}{Third-party cookies}
In the Web community, it has been established a mechanism for a Server to know that a Client had already contacted it before, in order to provide better answers to its queries. This mechanism is based on the so called \emph{cookies}, and it basically stores in the Client side some information the Server wants to keep about the Client for future reference. In later requests, the Client automatically sends that information to the Server, and then this latter can provide a customised answer.

However, a twist has been made to this mechanism, with which \emph{partners} of the Server can also store information in the Client side, without this latter even having visited the partner webpage. The Server is involved in this procedure, and it basically stores what it is called as \emph{third-party cookie}, with information about its partner (commonly advertisers that pay the Server to do that). Later, when the Client connects to some other Server, which has also an agreement with the same third-party, the Client will automatically send the information stored in the cookie to the third party, as this is automatically done whenever is requested from the Server side and the Client has the data.

In this Subsection, we abstract this behaviour and formalise it in our \textbf{AspectKBL} language, in order to prove the Global Property that the Client will be secure to avoid receiving any third-party cookie.

\begin{subsubsection}{Our abstracted desciption}
When a Client establishes a connection with a Server, this latter one sets in the Client some information called ``cookie'' that the Client has to resend every time he wants to establish further connections with the same Server for a given period of time. Again and again, the information contained in a cookie for a given Server can be increased, due to further connections and new pairs name-value set by the Server inside the Client.

In some cases, Servers can set into the Client cookies ``owned'' by another website (mainly for advertisement purposes). The website that owns a cookie is the one that will receive it from the Client every time the Client connects to it. Therefore, if a given website ``domain1'' has some advertisement from a Company Ad, and another website ``domain2'' has the same type of advertisement, when Client connects to domain1 he receives the cookies set by domain1, but also a cookie set by domain1 but owned by Ad. Later, when Client connects to domain2, because some of the content comes from Ad, it is like establishing a ÒparallelÓ connection to Ad, and then the cookie is indeed sent to Ad, and this one learns that Client had previously connected to domain1 (before this current connection to domain2).

These mechanisms are depicted in Figure \ref{fig:thirdpartycookie}, showing the insecure operation and its further consequence, which leaks information and prevents privacy.

\begin{figure}[ht]
  \centering
  \includegraphics[width=1\textwidth]{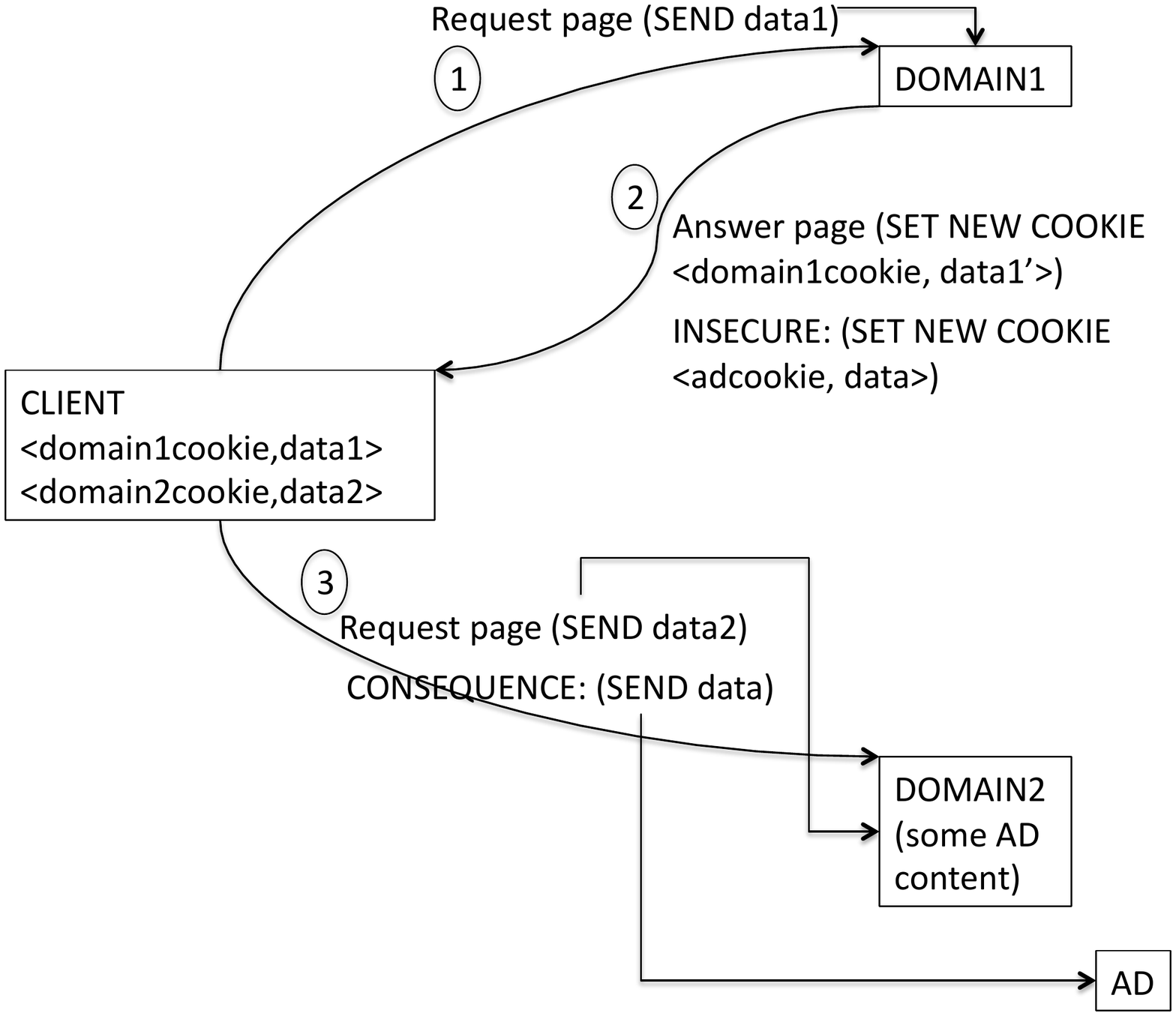}
  \caption{Third-party cookie abstracted.}
\label{fig:thirdpartycookie}
\end{figure}

To avoid this, one can set in the Client side, a security policy that does not allow third-party cookies. In the Web community this is known as the \emph{Same-origin policy}, and it is part of the Platform for Privacy Preferences Project (P3P).
\end{subsubsection}

\begin{subsubsection}{Formal model}
When a Client connects to a Server, either it is the first time and so no cookie is stored, so the Client just connects sending nothing, or perhaps it is not the first time and then some cookiedata must be sent. For this, the Client reads the data from its own storage, and then sends it to the Server, appending the possibly new data received for future use:
$$
\begin{array}{llll}
Client::  & \multicolumn{3}{l}{(\ \textbf{out}(\texttt{connect}, self, \texttt{nocookie})@Server\ .\ \emptyset} \\
& & + & \\
& \multicolumn{3}{l}{\textbf{read}(\texttt{Server}, !cookiedata)@self\ .} \\
& \multicolumn{3}{l}{\textbf{out}(\texttt{connect}, self, cookiedata)@Server\ .} \\
& \multicolumn{3}{l}{APPENDNEWCOOKIE\ )} \\
\end{array}
$$
The appending procedure has to be done stepwise, as our primitives only allow us to read or write, and erase while reading. Therefore, the Client reads the two pieces of cookies it has, erasing them, and then writes a new one with the combination of both:
$$
\begin{array}{ll}
APPENDNEWCOOKIE \equiv & \textbf{in}(\texttt{Server}, !data1)@self\ . \\
& \textbf{in}(\texttt{Server}, !data2)@self\ . \\
& \textbf{out}(\texttt{Server}, data1.data2)@self\ .\ \emptyset
\end{array}
$$

After the Server has received a connection, it decides to serve it by taking it from its space of pending connections. It will take the Client information and the cookie, and after that it might send back to the same Client some extra information for an updated cookie:
$$
\begin{array}{ll}
Server:: & \textbf{in}(\texttt{connect}, !Cli, !cookiedata)@self\ . \\
& \textbf{out}(self, \texttt{morecookiedata})@Cli\ .\ \emptyset \\
\end{array}
$$

On the other side, if the Server is one that tries to send third-party cookies, it will behave slightly different:
$$
\begin{array}{ll}
BadServer:: & \textbf{in}(\texttt{connect}, !Cli, !cookiedata)@self\ . \\
& \textbf{out}(\texttt{thirdparty}, \texttt{morecookiedata})@Cli\ .\ \emptyset \\
\end{array}
$$

The security policy that we could attach to the Client in order to avoid receiving such cookies from the bad latter Server, but still receiving those from the former one, is the following:
$$
\begin{array}{l}
pol_{EHDB} =
\Aspectbegin
\#s = \texttt{Server}
\\
\textif \#s::\textbf{out}(\texttt{Server},-)@self :
\\
\TrueS
\Aspectend[]
\end{array}
$$

Now, with the entire system, no matter which of the two Servers with put in parallel with the Client, we should be able to prove the following Global Property, which establishes that whenever we traverse a transition in which someone is writing a pair with the name of a Server and some extra data to the Client location, then that someone is indeed the very Server:

\begin{equation}
\label{eq:obligcookie}
AG_{\{\$s:\textbf{o}(\texttt{Server},\$data)@Client \}}\$s = \texttt{Server}
\end{equation}
\end{subsubsection}
\end{subsection}

\end{section}

\begin{section}{Conclusion}
\label{sec:conclusion}
In this report, we have presented a framework with which it is possible to describe distributed systems consisting of distributed processes and memory location. Furthermore, the possibility of attaching security policies to the locations provides access control methods for the interactions among them.

With that, one can induce a Labelled Transition System (LTS), and over it one can perform model checking tasks. In this work we have seen how to perform this tasks in an alternative way, using a twist of the framework and a particular version of Action Computation Tree Logic. We have proposed an algorithm that smartly model checks the systems developed within the framework, and provides guarantees about the desired global security properties. Everything was illustrated with detailed examples.

Our alternative approach to model check relies on some properties given by the framework in order to guarantee conditions that will hold in all executions, without the need to induce the entire LTS. With this, the model checking is performed in a fast and safe way. We discussed the limitations of our approach, and pinpointed some weak points or restrictions we are having right now. The development of a tool is described, and its implementation is left for the Appendix.

With this work, we aim at showing that a global security property could indeed be achieved by localising partial access control policies throughout the system, keeping everything well-organised, and without having the need of any central controller. The advantages of this are the possibility of more realistic implementation of security policies, as the policies are located where they can indeed be controlled by the implementor. We would allow, then, to design distributed systems where the policies attached to the different parts of the system can indeed be specified and analysed, and if the analysis gives the expected results, then the implementation of them could be straightforwardly done.

If this work comes through, the possible applications of our approach go even beyond the security domain, as we envision the possibility of combining policies for decision support systems, and any other situation where possible conflicts or unknown information may occur, while aiming to a certain global decision.

\begin{subsection}{Future work}
Throughout this document, some parts were mentioned to be partially done or missing. In this Subsection they are clarified and explicitly recalled, and some other ideas for future work are also pointed.

Our language is Turing-complete, though for doing our analysis we syntactically restrict it to have only finite paths, by forbidding the use of the replication operator $*$. In the future, we should consider the entire language, and find some necessary or sufficient conditions for our analyses to remain valid.

Another restriction we are having right now, is that the target location of the $cut$ in the obligations must be a constant. This allows us to fix the policies involved in the granting of the relevant actions. We might relax this restriction, allowing the use of variables, and in the cases were the action coming from the net is also a variable, so no grounding is possible before starting to consider the policies, we should then use some Static Analysis in order to restrict the set of possible policies. We then would need to do the analysis for each policy of the set separately, in order to see if with the entire set we can certify the action, or otherwise detecting for which policies (thereby identifying for which values of the target location) the action cannot be certified. By doing this, we could then iteratively improve the precision of the Static Analysis, to see if those values are indeed possible or they were considered due to the first over-approximation taken.

Staying in the imprecision assessment, we might also come up with examples that are secure but are not certified by our analysis, thereby showing that it is indeed an over-approximation. Alternatively, we might be able to prove that our analysis is actually precise. We have some clues, but this has to be further developed. Certainly, if we find out that our analysis does some over-approximation, having the examples could help us to define ways to give more precision to it.

On the opposite side, as we do not want our analysis to be an under-approximation, as we want it to be safely certifying actions, we need to prove that this is indeed the case. Certainly, to be sure that every network that is certified by our analysis is indeed secure, a proof of correctness is necessary. This proof should show that the way we do our alternative approach of model checking (including the inversion in the order the substitutions are taken) is correct with respect to the Semantics of our language.

We have illustrated our approach for simple examples, and indeed they can certainly be currently solved. However, although we have mentioned the possibility of coming up with sets of constraints, given by the recommendations of the policies that might allow a certain action, therefore being able to prove the implication from them to the predicate we might be analysing, we are actually not doing this for the moment. Certainly, the possibility of using some constraint solver for pursuing this task is being considered. For the moment, we are only taking the power of the substitutions done using our approach, and if with them we are able to certify the action, then we do it. Otherwise, we just answer that the action cannot be certified, although it might be just a matter of Logically understanding/arranging the certain constraints for certifying it.

Considering that, we estimated that in some case it could still be impossible to come up with a definite answer that some action is not certified. Indeed, even if the set of constraints is properly considered and manipulated, if we do not have some constraint that is opposite to the predicate we are aiming at proving, then we should not completely rule the action out, as we are not sure whether it will perform some insecure behaviour according to our aims. Therefore, assessing the possibility of having a 3-valued answer could be a direction to follow. A certainly certified network would thereby obtain the answer $\TrueB$, and a network that is clearly insecure would obtain the answer $\FalseB$. There might be other networks that, due to some imprecision in our analysis, or to some complicated structures that are beyond the power of the current state of the art of our tool, can obtain the answer $Unknown$.

Using the examples we have described in Section \ref{sec:examples}, we shall come up with improvements to our theory, and show its usefulness as well. On one side, aiming at detecting some other insecurities in the examples, both known ones and hopefully also some unknown one, would be one good goal. For doing this, it would be necessary  to lower the abstraction level and give more details to the examples. If we can achieve this, then the usefulness of our approach will be further assessed. On the other side, we can use the examples to have an \emph{example-guided theory improvement}. Indeed, we aim at having more powerful Logics for the global properties, and instead of just coming up with ideas that make our Logics closer to some existing one, we aim at using our examples to come up with useful properties that would be nice to describe about them, and extending our Logics towards being able to express, and to later analyse, such properties. In particular, the possibility of considering \emph{liveness} properties as well, as opposed to safety properties as we are doing right now, is on our agenda. Moreover, the continuation processes shall be considered as well, as they are indeed an important part of our policies, and since we are working with a process algebra they are indeed first class objects, so it would be very interesting to make good use of them for expressing properties.

Finally, and moving specifically to the implementation of the tool, not just the constraint solving part is missing. Indeed, the fact that the policies that are relevant for a certain action are combined using the liberal approach is indeed considered from the basic implementation of the tool. So the possibility of dealing with other approaches could be assessed, and some modifications of the tool shall be done in order to make it more flexible. Moreover, the combination of policies coming from a single location could follow any Belnap operator, as that depends on how they are attached to the location. Indeed, dealing with this in a more abstract way is also necessary, as currently they are combined by collecting their answers and finding a final one also considering that $\bot$ would finally allow and $\top$ would not. Certainly, having two levels of abstraction and keeping in mind within the tool that we are dealing with 4-valued elements is necessary, in order to, only in the end, map it into 2-valued Logic.
\end{subsection}
\end{section}


\bibliographystyle{plain}
\bibliography{bibtex}


\appendix
\begin{section}{Appendix: Implementation of a tool}
\label{sec:appendix}
This Appendix describes with some degree of detail how a tool is being built to deal with the alternative approach to model checking described in Section \ref{sec:SmartModelChecking} for the framework discussed in this whole report. Some algorithm design issues are discussed in high level first, then how to deal with some specific issues, and finally some Haskell implementation parts are shown.

\begin{subsection}{The tool}
\label{subsec:toolpseudocode}
Currently, the tool does not perform the constraint analyses described in the algorithm part of Subsection \ref{subsec:smartmodelcheckingformal}. It indeed assumes the recommendations and the predicates in the obligations are straightforward enough that with the power of the substitutions found it can be determined whether the policies imply the obligation. It is the topic of near future work to improve this restriction.

For the moment, the algorithm described in Subsection \ref{subsec:smartmodelcheckingformal} can then be implemented in pseudo-code as in Table \ref{tab:algorithm_modelchecker}, omitting the details of how the \texttt{mightgrant} is implemented.

\begin{table}[t]
\begin{verbatim}
Algorithm: check [network] [obligation]

for each action [act] occurring in some process at some location of [network]
  if [act] might match the cut of [obligation] producing substitution [theta0]
  then
    if [act] substituted by [theta0] might be granted by the involved policies
    then
      if predicate of [obligation] substituted by [theta0] might be False
      then
        Return False /* the obligation may not be satisfied */
      else
        Return True /* the obligation is satisfied */
    else
      Return True /* the obligation is satisfied */
  else
    Return True /* the obligation is satisfied */
\end{verbatim}
\caption{Algorithm for smartly model check.}
\label{tab:algorithm_modelchecker}
\end{table}

The substitutions that are done for matching the $cut$ of the obligation with the action, and the $cut$ of the involved policies with the action, are not then applied in the same places. Indeed, the substitution obtained with the obligation is the one that makes us know in which situations the action is relevant for our obligation. Therefore, such relevance is then considered while applying the substitution to the action before checking if the involved policies would allow it. The substitutions that might appear from the policies are only applied within the policy, to check the $cond$ and the $rec$. Although it might seem strange, since the policies are the actual granters of actions thereby generating transitions for the LTS which later could be matched by the obligation, this is a proper approach, since the $rec$ is the actual influencer of whether the predicate could be satisfied. This shall be formally proved in future work.

Besides, the unification procedures are straightforwardly implemented according to Subsection \ref{subsec:smartmodelcheckingformal}. The \texttt{mightgrant}, in its turn, is implemented stepwise by analysing the different parts of each involved aspect, and then a combination of them is done using Belnap Logic, for finally obtaining which the decision might be. In cases where no definite decision can be estimated, it is very difficult to deal with without using Static Analysis, as for instance if some recommendation might be $\TrueB$ under certain conditions, but it might also be $\FalseB$, then the negation of that recommendation might be $\TrueB$ of course. Then, either with or without negation nothing can be decided, due to some variables that are not ground.
\end{subsection}

\begin{subsection}{Dealing with implementation}
The tool implements what is stated in the previous Subsection. For this, a Haskell prototype is being developed, which currently consists of three main modules: \texttt{DataStructures.hs}, \texttt{Substitution.hs} and \texttt{Algorithm.hs}. They (and some other auxiliary ones) are organised according to Figure \ref{fig:module_structure}.

\begin{figure}[ht]
  \centering
  \includegraphics[width=1\textwidth]{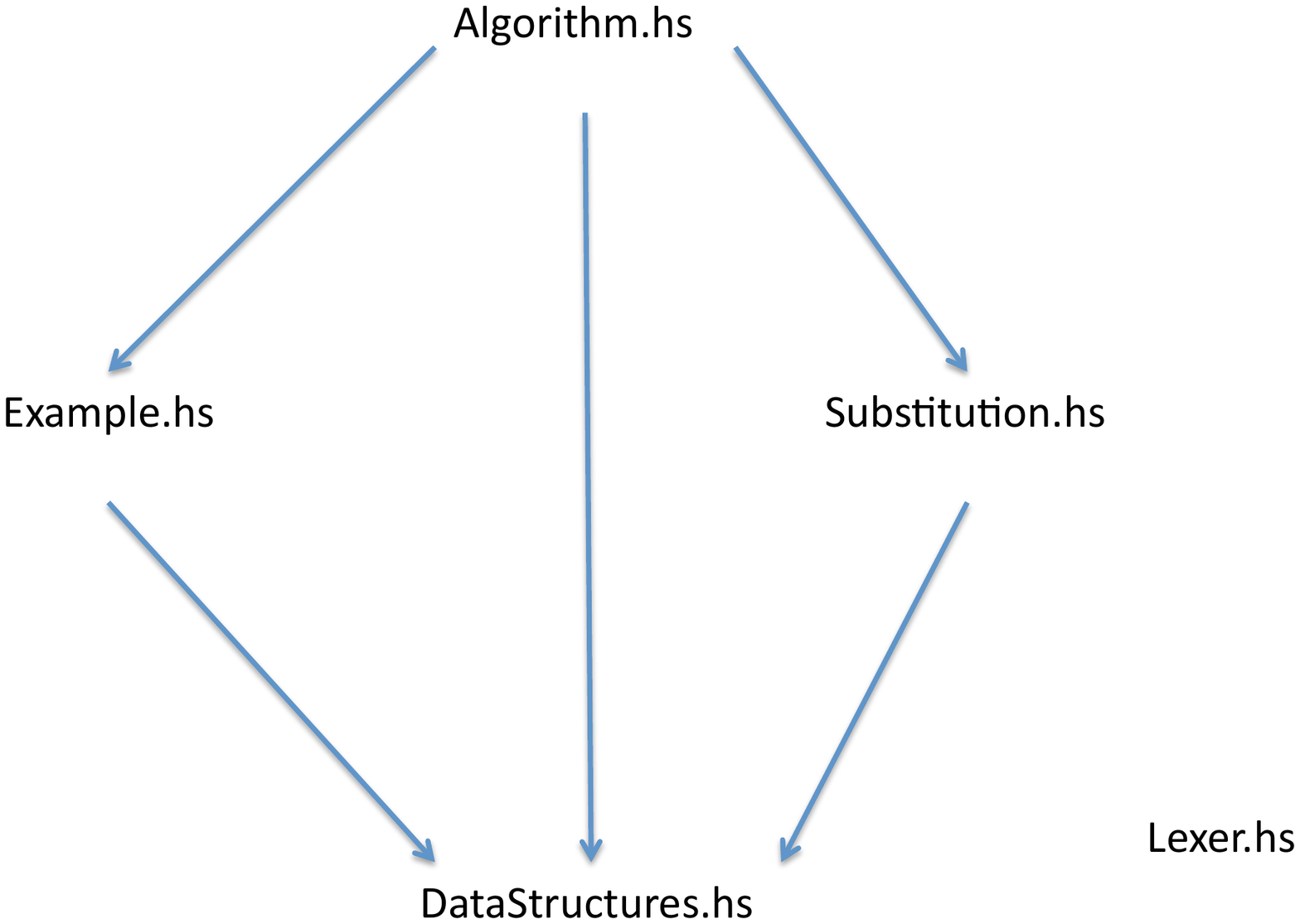}
  \caption{Module structure of tool.}
\label{fig:module_structure}
\end{figure}

The first one is in charge of defining the internal structures that will be used for keeping the abstract syntax tree of the examples the tool has to deal with. It mainly follows the abstract syntax of the formal AspectKBL language, with some slight modification due to tags or so. It also defines the internal structures for keeping the global properties, and some auxiliary functions for accessing the structures. In particular, one of the main functions is \texttt{takelactions}, which is in charge of obtaining the list of all the actions that have to be certified for security, by recursively entering all the locations in the network. The actions are listed including the location that holds each of them, as this is a necessary piece of information that is needed, because it defines some of the policies that will govern the action.

The pretty printing of these structures is also (for the moment just partly) defined in this module. The structures of Augmented Locations, which will later be used by the substitutions, is also defined here (perhaps the wrong place, though more easily findable during development process).

Some of the structures of the language are not completely defined. In some cases this is just due to lack of time and to incremental building of the tool, so they are appointed to be developed soon, for instance combination of recommendations within a single aspect, or quantifiers in the predicate of the global property. On the other hand, there are structures that are not completely defined due to lack of theory on how to deal with them when doing the over-approximation. For instance, if an aspect ÒmightÓ trap an action, and in such cases the action is granted, then we could check the implication from the recommendation to the predicate of the global property. However, if the aspect is modified to have the negation of the condition (or the recommendation), it ÒmightÓ still be trapping, because with the non-negated we are not sure it always traps (it was also a might and not a must). Maybe 4-valued logic could also be used to deal with the over-approximation, or probabilities (not really considered), or really distinguish when the analysis has completely grounded variables and so each conclusion can be ÒdefiniteÓ instead of ÒpossibleÓ.

The binding of variables is also another issue that is still to be clearly established how to implement.

The second module of the tool, \texttt{Substitution.hs}, relies on the module \texttt{DataStructures.hs}, and is in charge of defining the necessary functions for obtaining a substitution that unifies labels ($cut$ from global properties) and actions or pointcuts ($cut$ from aspects) and actions. It also defines the functions for applying the substitutions to any possible entity that is plausible to have variables (action, processes, networks, policies, conditions, recommendations, predicates, etc.).

A naming convention for the variables is used within the prototype, for the internal handling of substitutions and avoiding name crashes or free variables. Variables defined in a process could have any name, variables defined in an aspect/policy must start with the character hash (`\#'), and variables defined in an obligation must start with the character dollar (`\$'). For instance, if a variable occurs in an obligation without the character dollar, it means it might be a free variable, or referring to a variable of the network, and this is strictly forbidden. Anyway, the tool is capable of dealing with these, because for internal computations while applying substitutions it might be some points where this is the case. That is why one very important point for using the tool is to make sure that when calling the main function, this does not happen.

Finally, the module \texttt{Algorithm.hs}, which relies in both \texttt{Datastructures.hs} and \texttt{Substitution.hs} modules, defines the model checker itself, the function that given a network and an obligation (a global property) will try to certify it. It basically makes use of the function \texttt{takelactions} (defined in the module Datastructes) and then checks each of them separately. The network is certified if and only if all of them can be certified.

For certifying an action, it implements the procedure described in the previous Subsection, of first looking for a substitution that could unify the action with the label of the obligation. If no substitution is found then the action is automatically certified. If a substitution is found, then it checks for each and every aspect that will govern the action if it might grant the action. If at least one aspect will never grant the action, then the action is certified. Otherwise, if the action might be granted, then it is checked whether the predicate might be False in some cases (and as established in the previous Subsection, for the moment no system of constraints is used for this, but it is just done by relying on the power of substitutions). If the predicate will never be False, then the action is certified.

If in the end the action could not be certified, the answer of the model checker is False, but this might mean the network is not secure, or it might also mean that the procedure was not precise enough and no certifying way was found.

The tool has also implemented a module \texttt{Lexer.hs}, which appears to work properly (but currently just for a subset of the language), but it cannot be used yet because it is also necessary another module Parser, which has not been implemented due to lack of time and concentration in other issues. One very important point that will have to be included in the parser, is that for every variable occurring in the obligation an starting with a dollar symbol (`\$') must be added, and same for variables in the aspects with the hash symbol (`\#'), due to the issue discussed before. This has to be done while constructing the abstract syntax tree.
\end{subsection}

\begin{subsection}{Some code}
In the module \texttt{Substitution.hs}, the necessary functions for finding and performing substitutions are implemented. The unification functions are almost straightforwardly implemented following the formal model of Subsection \ref{subsec:smartmodelcheckingformal}, except for some little modifications to make them fit in Haskell code and type system.

Moreover, the main function in charge of finding substitutions is \texttt{findsubs}, which indeed works for any pair of entities, but indeed we are interested in the cases of $cut$ from the global property against action, and $cut$ of aspect against action. This first step is done by ad-hoc functions that split each entity into parts and get rid of the type of the $cut$, so we are able to call the function \texttt{findsubs}. This latter is defined as follows:
\begin{verbatim}
findsubs :: Oper -> Oper -> AugLoc -> AugLoc -> [AugLoc] ->
    [AugLoc] -> AugLoc -> AugLoc -> AugLoc -> Maybe Substitution
findsubs op1 op2 alsrc1 alsrc2 als1 als2 altgt1 altgt2 ignored =
    if theta3 /= Nothing
    then Just ((fromJust theta1) ++ (fromJust theta2) ++
        (fromJust theta3))
    else Nothing
    where
        theta1 =
            if op1 == op2
            then unifywithignored alsrc1 alsrc2 ignored
            else Nothing
        theta2 =
            if theta1 /= Nothing
            then
                if length als1 == length als2
                then unifyseqwithignored
                    (applysubs (fromJust theta1) als1)
                    (applysubs (fromJust theta1) als2)
                    ignored
                else Nothing
            else Nothing
        theta3 =
            if theta2 /= Nothing
            then unifywithignored
                (applysubs (fromJust theta2)
                    (applysubs (fromJust theta1) altgt1))
                (applysubs (fromJust theta2)
                    (applysubs (fromJust theta1) altgt2))
                ignored
            else Nothing
\end{verbatim}

Once a substitution is obtained, it can be used for actually substituting any entity defined within the network. Indeed, this is already illustrated even within some parts of the \texttt{findsubs}. The function is a polymorphic one, and its main definition is as follows:
\begin{verbatim}
class Substituible a where
    applysubs :: Substitution -> a -> a
\end{verbatim}
Then, it has a specific implementation for each type of entity we are interested in making it \texttt{Substituible}, and certainly all those defined within our framework have their implementation.

\begin{paragraph}{The Haskell algorithm}
Finally, the algorithm for performing the smart model checking is implemented in Haskell following the pseudo-code from Subsection \ref{subsec:toolpseudocode}. It still misses some constraint solving routine to deal with complex cases, but for simpler ones it makes use of the power of substitutions to come up with an answer in a safe way. The code is the following:

\begin{verbatim}
checksingleaction :: Label -> Predicate -> Network -> LocAction -> Bool
checksingleaction lab pred net lact =
    if theta0 /= Nothing    -- label might match action
    then
        if mightbegranted    -- policies might grant the action
        then
            if predicatemightbefalse    -- predicate might be false
            then False
            else True
        else True
    else True
    where
        theta0 = findsubslabel lab lact
            -- substitution that makes the action match the label
            -- of the obligation
        lact0 = applysubs (fromJust theta0) lact
            -- if action matched the label, then it must be substituted
            -- to evaluate the rest
        net0 = applysubs (fromJust theta0) net
            -- this might be useless, but just to have some standard,
            -- the network is also substituted by the theta0 found
        pol1 = (polsrc net0 lact0)
            -- this is the policy governing the action,
            -- coming from the source of the action
        pol2 = (poltgt net0 lact0)
            -- this is the policy governing the action,
            -- coming from the source of the action
        mightbegranted =
            if and [mightgrantpol1, mightgrantpol2]
            then True
            else False
        mightgrantpol1 = mightgrant pol1 lact0 net0
        mightgrantpol2 = mightgrant pol2 lact0 net0
        predicatemightbefalse =
            -- HERE is the place where the constraints should be applied,
            -- instead of just relying on the power of the substitution
                if evaluatealwaysTrue (applysubs (fromJust theta0) pred) net0
                then False
                else True
\end{verbatim}

Certainly, that is the code for checking whether a specific action occurring in the network satisfies the global property. For the whole network to satisfy it, all its individual actions must do so, as observed in Subsection \ref{subsec:smartmodelcheckinginformal}, but this is a simple mapping in Haskell and it is performed by the following code:
\begin{verbatim}
checkexample :: Network -> Obligation -> Bool
checkexample net obl = and xs
    where
        xs = map (checksingleaction lab pred net) tls
        lab = label obl
        pred = predicate obl
        tls = takelactions net
\end{verbatim}

Finally, for knowing whether a security policy might grant an individual action according to the information we have so far (just substitutions to make the action relevant for the global property), we first divide the policy into single aspects whose decisions are later to be combined using Belnap Logic, and apply the following function to each of them:
\begin{verbatim}
mightgranttrapping :: Aspect -> LocAction -> Network -> Bool
mightgranttrapping (Advice rec cut cond) lact net =
    if theta /= Nothing
    then
        if conditionmightapply (applysubs (fromJust theta) cond)
        then
            if recommendationalwaysdenies (applysubs (fromJust theta) rec) net
            then False
            else True
        else True
    else True
    where
        theta = findsubspointcut cut lact
\end{verbatim}

Certainly, the implementation of functions such as \texttt{conditionmightapply} and \texttt{recommendationalwaysdenies} is done but for complex cases some extra theory must be developed, since as discussed for instance in Subsection \ref{subsec:toolpseudocode}, cases with negation or other operators that might change completely our ``might'' approximation, can leave us in a problem due to large imprecision.
\end{paragraph}

\end{subsection}

\end{section}

\end{document}